\documentclass[12pt, draftclsnofoot, onecolumn]{IEEEtran}

\usepackage{setspace}
\usepackage{graphicx}
\usepackage{amsmath}
\usepackage{amssymb}
\usepackage{mathtools}
\usepackage{amsthm}
\usepackage{algpseudocode}
\usepackage{algorithm}
\usepackage{makecell}
\usepackage{subcaption}
\usepackage{float}
\usepackage{adjustbox}
\usepackage[table,x11names]{xcolor}
\usepackage{colortbl}
\usepackage{fancyhdr}
\usepackage{optidef}

\UseRawInputEncoding 
\DeclareCaptionFormat{ruled}{\leavevmode\leaders\hrule height 0.8pt depth0pt\hfill\mbox{}\endgraf#1#2 #3 \vspace{-0.4\baselineskip}\leavevmode\leaders\hrule height 0.6pt\hfill\null\vspace*{-0.6\baselineskip}}
 \algrenewcommand\algorithmiccomment[1]{\hfill\({}\triangleright{}\){\footnotesize#1}}%

\newtheorem{theorem}{Theorem}[section]

\newtheorem{lemma}[theorem]{Lemma}
\newtheorem{remark}[theorem]{Remark}

\newcommand{\tr}[1]{\textrm{#1}}

\definecolor{change}{HTML}{0000FF}
\newenvironment{change}{\color{blue}}{}

\graphicspath{{Figures/}}

\begin{document}

\title{Energy Efficient Operation of Adaptive \\ Massive MIMO 5G HetNets}

\author{
Siddarth Marwaha,~\IEEEmembership{Student Member,~IEEE,} Eduard A. Jorswieck,~\IEEEmembership{Fellow,~IEEE,}
Mostafa Jassim,~\IEEEmembership{Student Member,~IEEE,}
Thomas Kuerner,~\IEEEmembership{Fellow,~IEEE,}
David Lopez Perez,~\IEEEmembership{Senior Member,~IEEE,}
Xilnli Geng, Harvey Bao%
}

\markboth{IEEE Transactions on Wireless Communications.}%
{Marwaha \MakeLowercase{\textit{et al.}}: Energy Efficient massive MIMO 5G Networks}

\maketitle

\begin{abstract}
For energy efficient operation of the massive multiple-input multiple-output (MIMO) networks, various aspects of energy efficiency maximization have been addressed, where a careful selection of number of active antennas has shown significant gains. Moreover, switching-off physical resource blocks (PRBs) and carrier shutdown saves energy in low load scenarios. However, the joint optimization of spectral PRB allocation and spatial layering in a heterogeneous network has not been completely solved yet. Therefore, we study a power consumption model for multi-cell multi-user massive MIMO 5G network, capturing the joint effects of both dimensions. We characterize the optimal resource allocation under practical constraints, i.e., limited number of available antennas, PRBs, base stations (BSs), and frequency bands. We observe a single spatial layer achieving lowest energy consumption in very low load scenarios, whereas, spatial layering is required in high load scenarios. Finally, we derive novel algorithms for energy efficient user (UE) to BS assignment and propose an adaptive algorithm for PRB assignment and power control. All results are illustrated by numerical system-level simulations, describing a realistic metropolis scenario. The results show that a higher frequency band should be used to support UEs with large rate requirements via spatial multiplexing and assigning each UE maximum available PRBs.

\end{abstract}


\IEEEpeerreviewmaketitle
\section{Introduction}

\footnotetext[0]{\thanks{Parts of this work have been published in \cite{Marwaha2022}.}} 
5G technology is engineered to deliver superior performance, offering gigabit data rates, ultra-reliable low latency, and massive connectivity to link everyone and everything \cite{9678321}. The advanced hardware and performance of 5G networks contribute to improved energy efficiency (EE) measured in bits per Joule, surpassing previous technology generations. However, despite this enhanced EE, the surge in data traffic, utilization of wider bandwidths, increased antenna count, and higher levels of base station (BS) densification in 5G currently lead to heightened energy consumption. This has resulted in an augmented carbon footprint of networks in the 5G era \cite{7448820}. To meet ambitious targets for net-zero greenhouse gas (GHG) emissions, it is crucial to design future wireless networks with a primary focus on enhanced energy efficiency.

\subsection{Motivation}

It has been reported in \cite{s19143126} that $57\%$ of the total network energy is consumed by the base stations (BSs), where power amplifiers (PAs), transceivers, baseband units, and cables account for approximately $65\%$ of the BS energy \cite{9678321}. Therefore, significant attention has been directed towards enhancing these modules to reduce BS energy consumption in recent years, leading to the proposal and investigation of various energy-saving techniques.

Energy-saving techniques are generally aimed at adapting the time, space, and frequency resources of the BS to traffic demands. Over the last 15 years, solutions to improve spectral efficiency (SE) in all these dimensions have been deployed. However, the wireless network does not always require maximum SE. Instead, the SE should be carefully chosen to satisfy Quality of Service (QoS) requirements while saving energy.

In the time domain, symbol shutdown turns off all PAs when it detects that a downlink (DL) orthogonal frequency division multiplexing (OFDM) symbol has no data to send, reducing BS power consumption without affecting QoS \cite{Tan22}.

In the spatial domain, massive MIMO is the key technology for improving the coverage, throughput, and reliability of 5G networks. The SE of massive MIMO-aided 5G sites is reported to be three to five times greater than that of 4G sites using traditional radio solutions. However, it is essential to select the best massive MIMO configuration according to traffic needs to avoid performance over-provisioning and maximize energy efficiency (EE). A good scaling of massive MIMO energy consumption can be achieved by switching off part of the massive MIMO PAs when traffic requirements are low \cite{NGMN21}, thus adapting to the load and QoS requirements to avoid coverage and performance degradation.

In the frequency domain, techniques including carrier aggregation and multi-connectivity have been developed to increase available bandwidth for enhanced data rates and/or reliability in 5G. As described in \cite{Guer11}, in heterogeneous networks (HetNets) operating multiple frequencies, carrier shutdown can periodically assess the service load of the multiple carriers managed by the BS. If the load in some capacity layers is lower than a specified threshold, all the massive MIMO PAs and the baseband processing can be deactivated, resulting in significant energy savings. However, QoS may be impacted if active base stations cannot manage the increased traffic resulting from carrier shutdown.

In this paper, we focus on the online configuration of hardware components in the massive MIMO BS in spatial and frequency domains.

\subsection{State-Of-The-Art}

For an overview of energy-efficient wireless communications and fundamental green trade-offs, such as spectral efficiency (SE) versus energy efficiency (EE), deployment efficiency versus EE, delay versus power, and bandwidth versus power, as well as energy harvesting for sustainable green 5G networks, see \cite{Buzzi2016}, \cite{Zhang2017z}, and \cite{Wu2017}.

While the overview in \cite{Li2020} delves into power-saving mechanisms for User Equipment (UE) in great detail, our specific focus is directed towards reducing network energy consumption under UE quality of service (QoS) requirements. 
Detailed insights into different network power consumption models, EE metrics, and main EE-enabling technologies provided by the third generation partnership project (3GPP) new radio (NR), as well as power-saving techniques in 5G NR, can be found in \cite{9678321}. 
Additionally, \cite{Masoudi2019} provides an overview of means to monitor and evaluate EE.

Numerous recent papers address single-cell EE bounds and trade-offs, with some providing analytical results and EE bounds and trade-offs with SEs. Efficiently solving typical EE optimization problems often involves the use of fractional programming (see \cite{Zappone2015} for an overview and references therein).
In multi-cell scenarios, the model is more challenging due to complex topology, UE distributions, traffic models, wireless fading channels, protocols, algorithms, and a large number of parameters. As a result, explicit and general closed-form expressions describing the EE bounds and trade-offs of a multi-cell massive MIMO network are lacking. In the following, we summarize selected recent state-of-the-art works where the energy consumption of multi-cell wireless networks is studied.

For massive MIMO systems, a simple linear or affine model as a function of the transmit power has been identified as inadequate because it leads to unbounded EE as the number of antennas \begin{change}grows\end{change} large \cite{6891254}. Therefore, in massive MIMO systems, it is important to incorporate power consumed by different BS components, such as PAs, transceivers, analog filters, and oscillators. Accounting for circuit power consumption, in \cite{6951974} and \cite{7031971}, EE has been shown to be a quasi-concave function of \begin{change}the\end{change} number of antennas, number of UEs, and the transmit power, where EE is optimized by increasing the transmit power with the number of antennas. In \cite{7248569}, it is advised to turn off a fraction of antennas to reduce the total power consumption late \begin{change}at\end{change} night when the traffic demand is low. Whereas in \cite{8094316}, the downlink EE is maximized by adapting the number of antennas to temporal load variations over a day. In \cite{article}, the authors \begin{change}adjusted\end{change} the number of antennas and transmit data rate to maximize EE for uplink energy-efficient resource allocation in very large multi-user MIMO systems. While in \cite{7438738}, the authors \begin{change}aimed to find\end{change} the optimal densified network configuration, concluding that reducing the cell size leads to higher EE; however, the EE saturates when the circuit power dominates over transmit power.

Instead of binary decisions to switch on and off complete BS, the activation of sleep modes with finer granularity levels \cite{Debaillie2015} is considered for EE optimization \begin{change}too\end{change}. In \cite{Liu2016a}, EE with the introduction of several levels of sleep depths is optimized. The BS density for enhancing EE through traffic-aware sleeping strategies in both one- and two-tier cellular networks is optimized in \cite{Li2016z}. In \cite{Shojaeifard2017}, the most energy-efficient deployment solution for meeting certain minimum service criteria is developed, and the corresponding power savings through dynamic sleep modes are analyzed. In \cite{Chiaraviglio2017}, the rate of failures triggered by fatigue processes of BSs subject to sleep modes is controlled by an algorithm called LIFE, applied to HetNets (LTE and legacy UMTS).

In \cite{Khodamoradi2020}, the power control and EE in downlink multi-cell massive MIMO systems are investigated and optimized. \cite{Dong2020} optimizes the underlying multi-user transport layer of the frequency division duplex orthogonal frequency division multiple access (FDD-OFDMA) massive MIMO system. The goal of \cite{Salh2020} is to maximize the non-convex EE in a downlink massive MIMO system using a proposed energy-efficient low-complexity algorithm. The optimization of the power assignment to achieve the maximum EE for a downlink system of single-cell massive MIMO based on a tight approximate expression for the achievable sum-rate is studied in \cite{Li2020w}. In \cite{9676676}, the authors jointly optimize the transmit beamforming vectors of femto base stations (FBSs) and the phase-shift matrices of reconfigurable intelligent surfaces (RISs) for an RIS-aided HetNet under channel uncertainties and residual hardware impairments to maximize the minimum EE of femtocells.

In addition, we refer the reader to \cite{7378276} for a detailed overview of UE to BS association techniques covering HetNets, massive MIMO, millimeter wave (mmWave), and energy harvesting and various UE association metrics, including EE. The authors in \cite{7417354} investigate the reassignment of UEs between adjacent small cells to enable spatial multiplexing gains via multi-user MIMO and reduce energy consumption through small cell sleep state. To maximize the EE of the system, the authors in \cite{7524485} investigate the optimal on-off switching and UE association in a HetNet with massive MIMO by formulating an integer programming problem. In \cite{he2015spectrum} and \cite{7153519}, the authors considered a comprehensive power consumption model for massive MIMO HetNets covering UE to BS association but did not consider PRB allocation.

While EE has been studied from various aspects, most works reported above do not consider multi-carrier transmission with spatial and spectral resource and component allocation\begin{change}, \end{change}and the matching of UEs to BSs and frequency bands.

\subsection{Contribution and Problem Statement}
In particular, our work, differently than previous ones, investigates the EE of a HetNet massive MIMO network, presenting and working with a more comprehensive power consumption model, which encompasses intricate factors such as UE to BS and frequency band association, PRB allocation, distributed power control, and the number of antennas/RF chains used. We propose a method that deconstructs a mixed-integer non-convex programming problem into sub-problems accounting for joint energy-efficient UE to BS assignment as well as carrier, PRB, and transmit power allocation. Our proposed method uses long-term channel statistics obtained from network level simulator SiMoNe (Simulator for Mobile Networks) \cite{Rose16}, which describes a realistic metropolis (Berlin) scenario with realistic antenna deployments.

Our analysis and results show that:
\begin{enumerate}
    \item 
    Under the assumption of homogeneous BS deployment and zero-forcing (ZF) with the same energy consumption model and with the same resource allocation strategy, there exists a simple sufficient condition to optimally decide on a BS carrier shutdown, and assign the UEs of the shutting down BS to its neighboring BSs. 
    \item 
    Under the assumption of fixed inter-cell interference and the use of ZF massive MIMO precoding with equal power allocation, there exists a sufficient condition to optimally allocate power to MIMO layers and PRBs. Whether a single or multiple spatial layers are used per PRB, the optimal power allocation requires using all available PRBs.  
\end{enumerate}

The rest of the paper is organized as follows. Section \ref{sec:sysmodel} introduces the system model, elaborating on the deployment (\ref{subsec:architechure}), signal (\ref{subsec:signalmodel}), channel (\ref{subsec:channelmodel}), power consumption (\ref{subsec:powerconsmodel}),  time and download (\ref{subsec:timedownloadmodel}) models used, and  
presents the problem statement (\ref{subsec:probstatement}). Section \ref{sec:mainresults} discusses the main theoretical results and the derived UE to BS assignment (\ref{subsubsec:analyticalresults} and \ref{subsubsec:alogs}) and resource allocation (\ref{subsubsec:analyticalresultsallocation} and \ref{subsubsec:alogsallocation}) algorithms. In Section \ref{sec:numericalresults}, 
the simulation setup (\ref{subsec:simusetup}) is described, and the numerical results (\ref{subsec:singlecellres} and \ref{subsec:hetnetres}) are presented. Finally, the conclusion and future work are summarized in Section \ref{sec:conclusion}.  

\section{System Model}
\label{sec:sysmodel}

\subsection{System Architecture}
\label{subsec:architechure}

Figure \ref{fig:model} shows an example realization of the system architecture, where a BS manages a single cell, which can operate different frequency bands, e.g., $700$ MHz band with $20$ MHz bandwidth and $2.6$ GHz band with $100$ MHz bandwidth. We highlight the same by showing a single BS in a cell, i.e., $2.6$ GHz band in blue or two BSs in a cell, i.e., $700$ MHz band in green and $2.6$ GHz band in blue. Each frequency band possesses a certain number of PRBs which can be assigned to the associated UEs.

\begin{figure}[htbp]
\centering
    \begin{subfigure}{0.6\linewidth}
    \centering
    \includegraphics[width=0.9\textwidth]{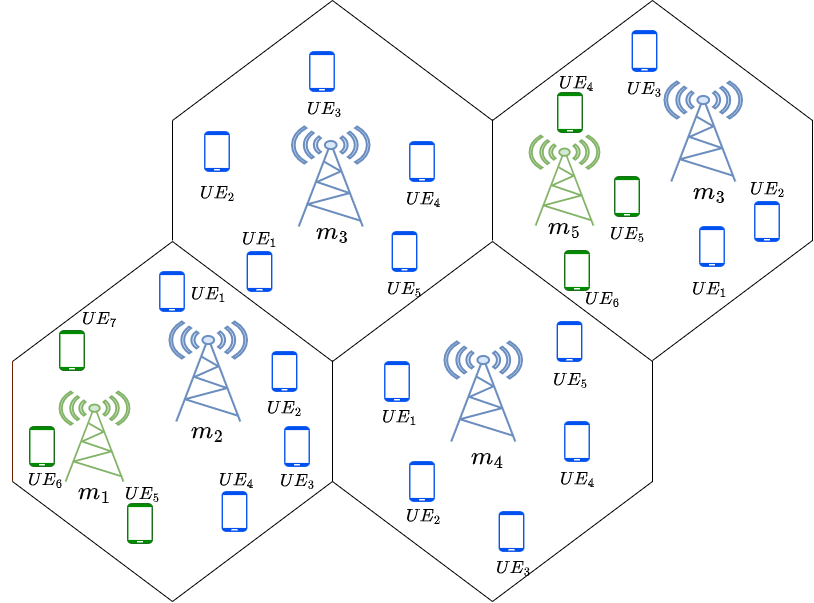}
    \caption{Macroscopic view: Multiple BSs, 
    each one managing multiple cells and serving multiple UEs.}
    \label{fig:multiplecells}
    \end{subfigure}
    \begin{subfigure}{0.3\linewidth}
    \centering
    \includegraphics[width=0.9\textwidth]{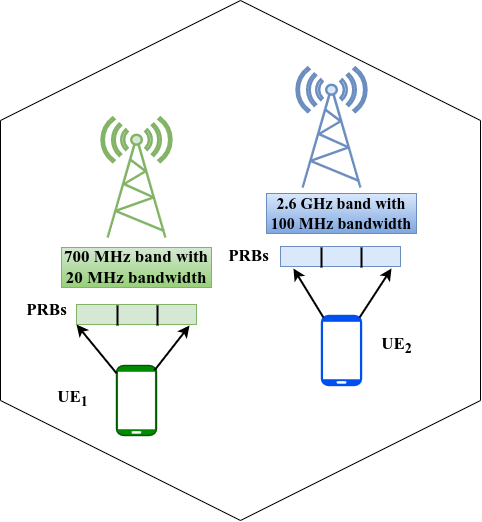}
    \caption{Microscopic view: Single BS managing two cells operating at different frequency bands, 
    where each frequency band has a given set of PRBs to assign to the UEs associated to the BS.}
    \label{fig:singlecell}
    \end{subfigure}
\caption{An example system scenario: 6 BSs consisting of one or two cells operating different frequency bands (e.g. $700$ MHz band in green and $2.6$ GHz band in blue.), 
where the UEs are uniformly distributed and the lower frequency band is sparsely deployed.}
\label{fig:model}
\end{figure}

In the system  model, $\mathcal{M} = \{1,\dots,M\}$ and $\mathcal{B} = \{1,\dots,B\}$ \begin{change}represent\end{change} the set of BSs and the set of available frequency bands per BS, respectively. $\mathcal{N}_{m}^{b}$ \begin{change}denotes\end{change} the set of available PRBs in the $m^{th}$ BS and $b^{th}$ frequency band, with cardinality $N_{m}^{b}$. 

The number of active PRBs at the $m^{th}$ BS in the $b^{th}$ frequency band is denoted by $\alpha_{m}^{b}$, i.e., $0 \leq \alpha_{m}^{b} \leq N_{m}^{b}$. Note that the triplet $\alpha = (m,b,n)$ \begin{change}represents\end{change} the PRB $n$ operated by the $m^{th}$ BS in the $b^{th}$ frequency band, and the functions $m(\alpha)$, $b(\alpha)$ and $n(\alpha)$ return the BS, frequency band and PRB index of PRB $\alpha$, respectively. The number of active transmit antennas at the $m^{th}$ BS in the $b^{th}$ frequency band is denoted by $a_{m}^{b}$, with a maximum of ${A}_{m}^{b}$ transmit antennas, i.e., $1 \leq a_{m}^{b} \leq A_{m}^{b}$. 

\begin{change}Furthermore\end{change}, $\mathcal{K} = \{1,\dots,K\}$ denotes the set of UEs. The assignment of the $k^{th}$ UE to the available PRBs in the $m^{th}$ BS and $b^{th}$ frequency band is realized by a matching function $\mu(k) = [ \alpha_1, ..., \alpha_{N_k} ]$,
which returns the set of PRBs assigned to the $k^{th}$ UE, with cardinality $N_k$. The matching function $\mu$ is overloaded,
and also returns the set of UEs assigned to a PRB $\alpha$, 
i.e.,
\begin{equation}
    \mu(\alpha) = \mu(m,b,n) = \mu(m(\alpha),b(\alpha),n(\alpha)) = \begin{cases}
        \mathcal{K}_{m,b,n} & \textrm{if UEs assigned} \\
        \emptyset & \textrm{if unassigned or not used}.
    \end{cases}
\end{equation}

\subsection{Signal Model}
\label{subsec:signalmodel}

We can distinguish between original massive MIMO beamforming models from \cite{Yang2013}, 
with recent application in \cite{Pramudito2020}, 
or the recent (new) massive MIMO beamforming models from \cite{Senel2019}. 

Typical forms of single-cell SINR expressions using ZF, often used in current networks, are\footnote{
It should be noted that the SINR expression in \cite{Yang2013}, \cite{Pramudito2020} and \cite{Senel2019} considers pilot contamination and intra-cell interference, which, however, are not accounted for in this work and hence the expressions have been adapted accordingly. 
However, pilot contamination can be included via lower channel gains or additional noise terms.} 
\begin{eqnarray}
    \tr{SINR}_k^{\tr{ZF}} & = & {(a-K) \beta_k p_k} \nonumber,
\end{eqnarray}
where $a$ is the number of active antennas, 
$K$ is the number of UEs, 
$\beta_k$ is the large-scale fading, 
and $p_k$ is the allocated transmit power. 
The multi-cell scenario studied in \cite[Section VI]{Senel2019} derived the effective signal to interference and noise ratio (SINR) for ZF as 

\begin{eqnarray}
    \tr{SINR}_{m,k}^{\tr{ZF}} = \frac{ (a_m - K) \beta_{m,k} p_{m,k}}{1 + \sum_{m' \in \mathcal{P}_m \setminus \{ m \}} (a_{m'}-K) \beta_{m',k} p_{m',k'}}, 
\end{eqnarray}
where $a_m$ is the number of active antennas at the $m^{th}$ BS, 
$\beta_{m, k}$ is the large scale fading of the $k^{th}$ UE at the $m^{th}$ BS, 
$p_{m,k}$ is the assigned transmit power to the $k^{th}$ UE by the $m^{th}$ BS, $\mathcal{P}_m$ is a set consisting of all the BSs, 
and $m'\in \mathcal{P}_m \setminus \{ m \}$ represents the set of BSs creating interference.

It is important to note that these SINR expressions are derived for independent and identically distributed (iid) Rayleigh small scale fading. 
In particular, the channel matrices of different UEs are assumed to be independent, and the number of antennas at each BS large enough. 
In \cite{Yang2017}, the case with a large number of antennas but line-of-sight (LOS) conditions is studied. In the LOS scenario, channel correlation among UEs can happen. By proper UE scheduling, and dropping highly correlated UEs, these cases can be avoided.  

In the following, we will consider the ZF achievable SINR expression from above along with the power and PRB allocation. Moreover, we will assume that a UE can only be associated to one BS and allocated to one frequency band, and thus, for the sake of simplicity, we partially remove the BS and frequency band indexes $m$ and $b$ in the following expressions.

In the downlink, the transmitted signal is generated at the BS by precoding and scaling the data symbols \cite{8641436}. Let $p_{k,\alpha}$ be the normalized transmit power applied on PRB $\alpha = (m,b,n)$  when serving the $k^{th}$ UE. Then, the data rate achieved by the $k^{th}$ UE over the set of assigned PRBs $\mu(k)$ is calculated using their average SINR\footnote{The SINR expression has been normalized by the noise power, and thus the noise variance has been included as part of the large scale fading. \begin{change}Additionally, the SINR expressions are computed under small scale Rayleigh fading. However, the large scale fading is later simulated via ray tracing.\end{change}} $\gamma_k$, as all these PRBs are coded together,
i.e.
\begin{equation}
    R_k = (\Bar{b}\cdot N_k) \log_2 \left( 1 + \frac{1}{N_k}\left(\sum_{\substack{\alpha \in \mu(k)}} \frac{ (a_{m(\alpha)}^{b(\alpha)} - |\mu(\alpha)|)\cdot p_{k, \alpha}\cdot \beta_{k, \alpha}}{1 + \sum\limits_{\substack{ \alpha': m(\alpha')\neq m(\alpha) \\ n(\alpha')=n(\alpha) \\
    b(\alpha')=b(\alpha)}} \sum\limits_{k'} p_{k', \alpha'}\cdot \beta_{k, \alpha'}}\right)\right), 
\label{eq: datarate}
\end{equation}
where $\Bar{b}$ is the PRB bandwidth\footnote{
The PRB bandwidth is $180$KHz in LTE and NR when using a sub-carrier spacing of $15$KHz.}, $a_{m(\alpha)}^{b(\alpha)}$ is the number of active transmit antennas at the $m^{th}$ BS in the $b^{th}$ frequency band, $|\mu(\alpha)|$ is the number of UEs spatially multiplexed on PRB $\alpha$, and $\beta_{k,\alpha}$ is the large scale fading experienced by the $k^{th}$ UE on PRB $\alpha = (m,b,n)$.

\subsection{Ray Tracing based Channel Model} 
\label{subsec:channelmodel}

We assume that the system operates in time division duplex (TDD) mode, where the receiver can directly estimate the channel information from the received signal \cite{4299599}, whereas, the transmitter relies on the reciprocity principle to obtain the channel information on the same flat-fading subcarrier \cite{6891254}. We consider ZF precoding for massive MIMO based on perfect channel state information at the transmitter (CSIT). Afterwards, we use the achievable rate expression obtained with channel hardening which only depends on the long-term channel gains for our algorithm design. 
    
Highly accurate channel estimates can be determined using ray tracing \cite{doi:https://doi.org/10.1002/9781118410998.ch10}. In this work, we use a ray tracing based large-scale fading predictor, called Femto Predictor (\textit{FemtoPred}), which was developed at the Institute for Communications Technology at TU Braunschweig. As input data, the ray tracer uses the following data sets, among others: building data, antenna diagrams, frequency of operation, subscriber locations. The ray search is then conducted using state-of-the-art ray tracing techniques and optimizations, as described in \cite{8739409}. Based on the results of the ray search, the influence of the free space path loss as well as that of other propagation effects such as diffraction, scattering, reflection and transmission are evaluated in order to determine each ray’s contribution to the overall received signal strength. The ray with the highest amplitude and phase shift at the receiver side was selected as the dominant ray. This ray, together with other selected paths, whose strength is larger than a threshold, contribute to the final received signal strength at the receiver side\footnote{The details for the large-scale fading computations can be found in Appendix \ref{sec:large_scale_fading}.}.

\subsection{Power Consumption Model}
\label{subsec:powerconsmodel}

The total power consumption across all the BSs and frequency bands considering load dependent and independent power consumption can be defined as \cite{Marwaha2022} \cite{6951974} \cite{9685216}, 
\begin{eqnarray}
    P_{tot} = P_{LD} + P_{LI}, 
\label{eq: egy_cons}
\end{eqnarray}
where the load dependent power consumption $P_{LD}$ is calculated as  
\begin{eqnarray}
    P_{LD} & = & \sum_{m \in \mathcal{M}} \frac{1}{\eta_{PA_{m}}} \sum_{b \in \mathcal{B}}\sum_{k \in \mu(\alpha)} \sum_{\alpha \in \alpha_{m}^{b}}  p_{k,\alpha},
\label{eq: LD}
\end{eqnarray}
and the load-independent power consumption $P_{LI}$ is computed as
\begin{eqnarray}
    P_{LI} & = & \sum_{m \in \mathcal{M}}\sum_{b \in \mathcal{B}} I_{m}^{b}\left(\frac{1}{\lambda_{m,0}^b}P_{m,FIX}^b + \frac{1}{\lambda_{m,1}^b}P_{m,SYNC}^b \right) +  \sum_{m \in \mathcal{M}}\sum_{b \in \mathcal{B}} a_{m}^b  D_{m,0}^b \nonumber\\ 
    & & + \sum_{m \in \mathcal{M}}\sum_{b \in \mathcal{B}} \left(a_{m}^b D_{m,1}^b \left(\sum_{k}N_k\right)   \right) 
    + C,
\label{eq: LI}
\end{eqnarray}
where $\eta_{PA_{m}}$ is the combined power amplifier and antenna efficiency of the $m^{th}$ BS,
$I_{m}^{b}$ is an indicator function, 
indicating if any PRB of the $m^{th}$ BS and $b^{th}$ frequency band is being used\footnote{
Note that, with this indicator function, 
we realise an ideal carrier shutdown in which the BS does not consume anything when no PRB is used.},
$P_{m, FIX}^b$ is the load-independent power consumption in the $m^{th}$ BS and $b^{th}$ frequency band required for site-cooling, control signaling, backhaul infrastructure and base-band processing, 
$P_{m, SYNC}^b$ is the load-independent power consumed by the local oscillator in the $m^{th}$ BS and $b^{th}$ frequency band, 
$D_{m,0}^b$ is the power consumed by the RF chain attached to an antenna in the $m^{th}$ BS and $b^{th}$ frequency band, 
including converters, mixers, filters, etc.,  
$D_{m,1}^b$ is the power consumed by the signal processing of a MIMO layer across a PRB in the $m^{th}$ BS and $b^{th}$ frequency band,
through $\lambda_{m,0}^b$ and $\lambda_{m,1}^b$ the power consumption model can capture a linear, sub-linear or an independent relationship between $P_{m, FIX}^b$, $P_{m, SYNC}^b$ and $D_{m,0}^b$, 
indicating the level of hardware sharing across the $B$ frequency bands in the $m^{th}$ BS, 
and $C$ is the fixed power consumed by the coding at a massive MIMO BS and backhaul support. 

\subsection{Time and Download Model}
\label{subsec:timedownloadmodel}

We assume a simple transmission model where ${X}_{k}$ bits of data are transmitted for each UE $k$, 
and the time ${T}_{k}$ to transfer this data is computed based on the achievable rate ${R}_{k}$ as
\begin{equation}
    T_k = \frac{{X}_{k}}{{R}_{k}}.
\end{equation}

Based on this,
the EE of the system can be defined as
\begin{eqnarray}
    EE = \frac{ \sum_{k \in \mathcal{K} } \frac{R_k}{T_k}}
    { P_{tot} }, \label{eq:EEm}
\end{eqnarray}
where $R_k$ and $P_{tot}$ are defined in (\ref{eq: datarate}) and (\ref{eq: egy_cons}), respectively. Furthermore, it should be noted that the UEs whose rate requirements cannot be satisfied, 
i.e. the outages, 
are not considered while computing the sum rate in the numerator of (\ref{eq:EEm}). 
Their contribution to the power consumption in the denominator of (\ref{eq:EEm}) is accounted for. The outage probability is reported in Sections \ref{subsec:singlecellres} and \ref{subsec:hetnetres} separately.

\subsection{Problem Statements and Preliminaries}
\label{subsec:probstatement}
In this subsection, we state the problem, under consideration and we collect some of the preliminary results from \cite{Marwaha2022} on the optimization of simple special cases of the scenario outlined above. 

If rate requirements of all UEs should be fulfilled with minimum transmit power, the following optimization problem is considered for all $m \in \mathcal{M}$ and $b \in \mathcal{B}$
\begin{subequations}
\begin{alignat}{2}
\min_{p, \alpha, \mu}   &\quad&     &  P_{tot} \label{eq:obj}\\
\text{subject to}     &\quad&     &{R_{k} \geq \underline{R}_{k} \label{eq:rc}} \\
                        &\quad&     & p_{k, \alpha}  \geq 0 \label{eq:pc} \\ 
                        &\quad&     &\sum_{k}\sum_{\alpha:m(\alpha)=m} p_{k,\alpha} \leq P^{\max} \label{eq:spc} \\
                        &\quad&     &0 \leq \alpha_m^b \leq N_m^b \label{eq:prbs_const} \\
                        &\quad&     &\max |\mu(\alpha)|< a_m^b \label{eq:ants} \\
                        &\quad&     &\left\{\bigcup_{n}\mu(m,b,n)\right\} \cap \left\{\bigcup_{n}\mu(m',b',n)\right\} = \emptyset, \forall (m',b')\neq(m,b) \label{eq:smb}
\end{alignat}
\label{eq:optprob}
\end{subequations}
where (\ref{eq:rc}) corresponds to the minimum rate constraints, (\ref{eq:pc}) to the non-negativeness of the power constraints, (\ref{eq:spc}) to the sum power constraints per BS and band, the maximum number of available PRBs constraints in (\ref{eq:prbs_const}), 
and the minimum number of required antennas constraints in (\ref{eq:ants}), and the last constraint (\ref{eq:smb}) means that the intersection between the the set containing all the users assigned to BS $m$ and the set containing all the users assigned to BS $m'$ must be an empty set. This excludes coordinated multi-point (CoMP) or multi-connectivity.

\textbf{Single User Scenario}:
As the first step, a single macro cell with one active frequency band serving a single UE is studied \cite{Marwaha2022}. The minimum rate constraint for a single UE is defined as, 
\begin{equation}
    \underline{R} \leq \Bar{b} \cdot |\alpha| \log_2 \left(1 + \frac{1}{|\alpha|} \sum_{n=1}^{|\alpha|} (a -1) \beta p_1\right),
\label{eq:rate_single_usr}
\end{equation}
where, $\underline{R}$ is the QoS requirement for the user, $\Bar{b}$ is the PRB bandwidth, $\alpha$ is the number of PRBs allocated to the user, $a$ is the number of antennas, $\beta$ is the large scale fading, and $p_{1}$ is the transmit power allocated to the user.  Assuming uniform power allocation for each PRB assigned to the user, i.e., $p_{1} = p$, the total power consumption of the macro cell can be expressed as, 
\begin{equation}
    P _{tot}  = \frac{\alpha}{\eta_{PA}}p + D_0 a + D_1 \alpha a
\label{eq:total_power_single_user}
\end{equation}
where, the constant $C$ is neglected for simplification and $p$ is obtained from  (\ref{eq:rate_single_usr}) as, 
\begin{equation}
    p \geq \frac{2^{\frac{\underline{R}}{\Bar{b}\alpha}} - 1}{(a - 1) \beta}.
\label{eq:power_single_user}
\end{equation}
To understand the behavior of the objective function, $P_{tot}$ is further analyzed and minimized with respect to number of antennas $a$ and number of PRBs $\alpha$ separately. It was observed that the optimum number of PRBs $\alpha^*$ and the optimum number of antennas $a^*$ can be computed by equating the derivative of the corresponding objective function with respect to $\alpha$ and $a$, respectively, to zero, which yields, 
\begin{equation}
    \alpha^{*} = \frac{\underline{R}\log(2)}{\bar{b}}\frac{1}{W\left({\frac{(D_1 \eta_{PA} a(a-1)\beta) - 1}{e}} \right) + 1}
\label{eq: opt_prb_single_user}
\end{equation}

\begin{equation}
    a^* = \sqrt{\frac{\alpha (2^{\frac{\underline{R}}{\bar{b}\alpha}} -1)}{\eta_{PA}\beta(D_0 + D_1 \alpha)}} + 1
\label{eq: opt_antennas_single_user}
\end{equation}
where, $W(\cdot)$ in (\ref{eq: opt_prb_single_user}) is the Lambert W function and $e$ is the Euler's number.

\begin{remark}
Depending on the parameters $\eta_{PA}$, $D_0$, $D_1$, $\beta$, and $\underline{R}$ there is an optimal number of active antennas. However, numerical evidence suggests that minimum number of antennas should be used as long as sufficient number of PRBs is available.
\end{remark}

\begin{remark}
Looking at the properties of the optimum $a^*$ and $\alpha^*$, we observe  that there must be an optimal number of PRBs to achieve a given rate while consuming lowest total power.    
\end{remark}

\textbf{Simple Multiuser Scenario:}
The total power consumption for a multi-user scenario can be computed as
\begin{equation}
    P_{tot} = \frac{1}{\eta_{PA}}\sum_{k=1}^K \alpha_k p_k + D_0 a + D_1 a K \sum_{k=1}^K \alpha_k.
\end{equation}
The following analysis considers a symmetric scenario, such that $\beta_1 = \beta_2 = \beta$, and $\underline{R}_1 = \underline{R}_2 = \underline{R}$ to compute the total power consumed with and without spatial multiplexing. For convenience, we start with this assumption and later relax it. 

\begin{lemma}[Lemma 1 \cite{Marwaha2022}]
Consider a symmetric scenario with $K$ users, where $\beta_1 = \ldots = \beta_K = \beta$ and $\underline{R}_1 = \dots = \underline{R}_K = R$. Then, the minimum power consumption is achieved for $\alpha_1 = \ldots = \alpha_K = \frac{\pi}{K}$, if $\pi$ is the total number of PRBs used. 
\label{lemma_symmetric}
\end{lemma}

\begin{theorem}[Proposition 1 in \cite{Marwaha2022}]
For a symmetric scenario, if the rate requirement $\underline{R}$ approaches small values, i.e. $\underline{R} \to 0$, the total power consumed when the users are not spatially multiplexed is lower than when the users are spatially multiplexed. Then, the optimal number of PRBs is one and minimum number of antennas should be used (for example $3$ antennas for $2$ users if spatially multiplexed or $2$ antennas for $2$ users if not spatially multiplexed)
\end{theorem}

\begin{theorem}[Proposition 2 in \cite{Marwaha2022}]
There exists a specific number of PRBs $\alpha^*$ for which $P_{tot}^{NSM} = P_{tot}^{SM}$. For all $\alpha < \alpha^*$, $P_{tot}^{NSM} < P_{tot}^{SM}$, whereas, for all $\alpha > \alpha^*$, $P_{tot}^{NSM} > P_{tot}^{SM}$.
\end{theorem}

\begin{remark}
The preliminary results above indicate that there exists a dichotomy between SM and NSM depending on the load of the cell. Furthermore, it seems more efficient to first fill up the spectral domain (allocating PRBs) than activating the next antennas and switch on spatial layers. The results in Section~\ref{sec:mainresults} also confirm this underlying intuition for multi-cell systems.     
\end{remark}

\section{UE Matching and Resource Allocation}
\label{sec:mainresults}
The programming problem in (\ref{eq:optprob}) is a mixed-integer non-convex programming problem which is difficult to solve jointly and globally. Therefore, we consider a divide-and-conquer approach. As illustrated in Figure \ref{fig:approach}, the workflow  starts with the generated (or measured) long-term channel parameters and rate requirements obtained from network level simulation. The first computational step is to perform the UE assignment to the BS. After the UE assignment is fixed, the component allocation (PRBs and spatial layers) together with power control is performed, assuming equal power is allocated across all PRBs assigned to a UE.

\begin{figure}[htp]
    \centering
    \includegraphics[width=.7\linewidth]{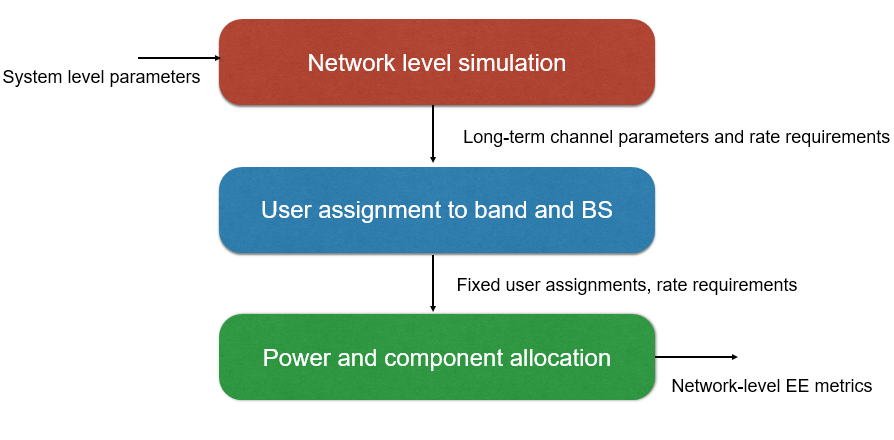}
    \caption{Approach to deconstruct the problem complexity into UE assignment and resource and component allocation.}
    \label{fig:approach}
\end{figure}  

\subsection{UE to BS Assignment}
\label{subsec:uebsassignment}
The baseline algorithm for UE assignment consists of choosing the band and BS combination which results in the minimum fading or maximum (channel) gain, i.e., 
\begin{equation}
    \mu(k) = \max_{\alpha} \beta_{k,\alpha}. \label{eq:bua} 
\end{equation}

In the following, the matching found by the basline algorithm is improved. 

\subsubsection{Analytical Results}
\label{subsubsec:analyticalresults}
The load-independent power $P_{LI}$ leads to the sub-optimality of (\ref{eq:bua}) when trying to solve (\ref{eq:optprob}). Furthermore, we conjecture that a lower number of UEs served by a BS corresponds to lower EE, because the constant power consumption dominates. In particular, BSs serving a single UE in one band can achieve a very low EE. Note that the EE of a BS can be defined as the ratio of the sum of rates delivered by the BS divided by the total consumed energy for delivering this rate, i.e., $E_m = \frac{ \sum_{k \in \mu(m)} R_k}{P_m}$, where, the \begin{change}power\end{change} consumed $P_m$ corresponds to the part of $P_{tot}$ in (\ref{eq: egy_cons}) for the one term related to BS $m$ in band $b$, i.e., 
\begin{eqnarray}
    P_m & = & \frac{1}{\eta_{PA_m}}\sum_{k}N_k \, p_{k,\alpha} +  I_{m}\left(\frac{1}{\lambda_{m,0}}P_{m,FIX} + \frac{1}{\lambda_{m,1}}P_{m,SYNC} \right) \nonumber \\ 
    & & + D_{m,0} a_m + \left(D_{m,1} a_m \left(\sum_{k}N_k \right) \right) 
    + C \label{eq:ecm}.
\end{eqnarray}
In the best case, the individual rates achieved should exactly correspond to the rate requirements, therefore, we could replace $R_k$ in (\ref{eq:EEm}) by the requirement $\underline{R}_k$. Assuming that the UE to BS assignment has been fixed and equal power is allocated across all the PRBs assigned to a UE associated with the BS $m$, it can be observed that only the first term in (\ref{eq:ecm}) is dependent on the transmit power allocated to the UEs. Therefore, for simplicity, the remaining terms can be dropped and a power minimization problem can be formulated subjected to minimum rate requirement, sum power budget and non-negative power allocation constraints. Analyzing the KKT conditions for the power minimization problem, it can be shown that the transmit power is assigned to each UE to maintain its minimum rate demand\footnote{The proof of the same has been provided in Appendix \ref{sec: replacement_proof}.}.

The idea for improving the UE to BS assignment consists in identifying the BSs that have the lowest EE, switch them off and distribute the formerly assigned UEs to the next best BSs and bands. 

\begin{theorem}
Under the assumption of homogeneous BS with the same load-independent energy consumption model, and with the same PRB allocation and power control strategy, it is energy optimal to switch off BS $m$ with $\ell$ UEs $\mu(m) = \{1,...,\ell\}$ assigned and assign them to their next best BSs $n_1,...,n_\ell$ if the following conditions are satisfied simultaneously
\begin{eqnarray}
    \beta_{k, \alpha}^{-1} + \Delta_l \geq \gamma_l \beta_{k,\alpha'}^{-1} \label{eq:cond}  
\end{eqnarray}
 for all $1 \leq l \leq \ell$, with $m(\alpha') = n_\ell$ and where $\Delta_l$ corresponds to the EE gain\footnote{$\Delta_l$ and $\gamma_l = \frac{\gamma_{2l}}{\gamma_{1l}}$ are computed explicitly for spatial multiplexing (SM) and no spatial multiplexing (NSM). However, due to space constraints, the proof is provided in Appendix \ref{sec: proof_theo1}.}.
\label{theo:1}
\end{theorem}

\begin{proof}
The proof is based on an inequality chain to bound the loss of EE of the new assigned BS $n$ compared to the EE gain of the switched-off BS $m$ and is provided in Appendix \ref{sec: proof_theo1}.
\end{proof}

The result in Theorem \ref{theo:1} motivates the algorithm to switch off BSs that have low EE, as described in Section \ref{subsubsec:alogs}. Furthermore, comparing the assignment of UEs for different bands, we observe that the SINR distribution over the UEs for different bands shows significant different support. This stems from the higher path losses for higher frequencies. The exact numbers are reported in Section \ref{subsec:simusetup}. The assignment rule in (\ref{eq:bua}) leads to sparsely filled higher frequency bands and crowded lower frequency bands. Therefore, the second idea is to improve the assignment in (\ref{eq:bua}) by adding a large-scale fading bias \cite{6166483} to the higher frequency channels, i.e., 
$$\tilde{\beta}_{k, \alpha} = \beta_{k, \alpha} + \Theta_{b(\alpha)},$$ where $\Theta_{b(\alpha)} \geq 0$ is the bias for the band $b$. The UE assignment is then based on the modified large-scale fading gains $\tilde{\beta}_{k, \alpha}$. However, the PRB assignment and power control is performed afterwards on the true large-scale fading channels $\beta_{k,\alpha}$.

\subsubsection{Algorithms}
\label{subsubsec:alogs}
The derivations above lead to four different algorithms which are implemented for the UE assignment. The first simple baseline algorithm is called the greedy UE assignment algorithm, as shown in Algorithm \ref{greedy}. It is the baseline scheme in (\ref{eq:bua}).

\begin{algorithm}
\caption{Greedy UE assignment}\label{greedy}
\begin{algorithmic}
\For{$i=1..K$}
    \State $\mu(k) \gets \max\limits_{\alpha} \beta_{k,\alpha}$. 
\EndFor
\end{algorithmic}
\end{algorithm}

\vspace{-1.0\baselineskip}
The second algorithm performs the re-matching of the UEs, which is implemented by removing the BSs with the lowest EE from the list of available BSs. Thereby, infinite loops, where a UE is re-matched to the second best (which might serve a single UE or multiple UEs) and then re-matching to the first one is avoided. The pseudo code can be found in Algorithm \ref{rematch}. 

\vspace{-1.0\baselineskip}

\begin{figure}[htbp]
  \begin{minipage}[t]{0.45\textwidth}
    \begin{algorithm}[H]
      \caption{Re-Matching UE assignment} \label{rematch}
      \begin{algorithmic}[1]
        \State \textbf{Run Algorithm \ref{greedy}}
        \While {$\exists m : |\mu(m)|\leq j $} 
        \State $k \gets \mu(m)$ ; $\mathcal{M} \gets \mathcal{M} \setminus m$
        \State $\mu(k) \gets \max\limits_{m(\alpha') \in \mathcal{M}} \beta_{k,\alpha'}$. 
    \EndWhile
      \end{algorithmic}
    \end{algorithm}
  \end{minipage}
  \hfill
  \begin{minipage}[t]{0.45\textwidth}
    \begin{algorithm}[H]
      \caption{Threshold UE assignment} \label{threshold}
      \begin{algorithmic}[1]
        \State \textbf{Run Algorithm \ref{greedy}}
        \While {$\exists m : |\mu(m)| \leq j $} 
            \If {$\frac{\beta_{k,\alpha} - \beta_{k,\alpha'}}{\beta_{k,\alpha'}} < \delta $} 
                \State $k \gets \mu(m)$; $\mathcal{M} \gets \mathcal{M} \setminus m$
                \State $\mu(k) \gets \max\limits_{m(\alpha') \in \mathcal{M}} \beta_{k,\alpha'}$. 
             \EndIf   
            \EndWhile        
      \end{algorithmic}
    \end{algorithm}
  \end{minipage}
\end{figure}

\vspace{-1.0\baselineskip}
Algorithm \ref{rematch} is further extended, where we propose a threshold based UE to BS assignment. In Algorithm \ref{threshold}, the UEs are reassigned to a BS only if the ratio $\frac{\beta_{k,\alpha} - \beta_{k,\alpha'}}{\beta_{k,\alpha'}}$, is less than a threshold  $\delta$ and the BS with no UE assigned are switched-off. Through the ratio ($\frac{\beta_{k,\alpha} - \beta_{k,\alpha'}}{\beta_{k,\alpha'}}$) and the threshold ($\delta$), we avoid assignment of the UEs to those BSs that provide weak channel gain resulting in low EE.



The fourth algorithm, as shown in Algorithm \ref{biased}, is the bias algorithm, where the UE to BS assignment is improved by adding a bias\footnote{Note that, as shown in \cite{6166483}, it is possible to formulate a large-scale fading bias model to chose a suitable bias value per BS, however, such a formulation further increases the complexity of the problem under consideration. Therefore, in this work, the bias is heuristically chosen by comparing the channel values of the $700$ MHz and $2.6$ GHz frequency bands. The exact value of the bias is provided in Section \ref{subsec:hetnetres}.} to the higher frequency channels. This not only enables the lower frequency band to be sparsely filled and the higher frequency band to be crowded, but it can also enable the selection of a better serving BS. 

\begin{algorithm}
\caption{Bias UE assignment}\label{biased}
\begin{algorithmic}
\For{$i=1..K$}
    \State $\tilde{\beta}_{k,\alpha} = \beta_{k,\alpha} + \Theta_{b(\alpha)}$
    \State $\mu(k) \gets \max\limits_{\alpha} \tilde{\beta}_{k,\alpha}$. 
\EndFor
\end{algorithmic}
\end{algorithm}

\vspace{-1.0\baselineskip}
Finally, the fifth algorithm is a combination of the bias algorithm with the re-matching, as shown in Algorithm \ref{biasrematch}, where first an artificial bias is injected and then the BSs are switched-off. Similarly, as shown in Algorithm \ref{biasthresh}, UE to BS assignment is also performed after injecting the bias and re-matching based on the threshold.    

\vspace{-1.0\baselineskip}

\begin{figure}[htbp]
  \begin{minipage}[t]{0.45\textwidth}
    \begin{algorithm}[H]
      \caption{Bias and Re-Matching UE assignment} \label{biasrematch}
      \begin{algorithmic}[1]
        \State \textbf{Run Algorithm \ref{biased}}     
        \While {$\exists m : |\mu(m)|=j$} 
            \State $k \gets \mu(m)$ ;  $\mathcal{M} \gets \mathcal{M} \setminus m$ 
            \State $\mu(k) \gets \max\limits_{m(\alpha') \in \mathcal{M}} \tilde{\beta}_{k,\alpha'}$. 
        \EndWhile
      \end{algorithmic}
    \end{algorithm}
  \end{minipage}
  \hfill
  \begin{minipage}[t]{0.45\textwidth}
    \begin{algorithm}[H]
      \caption{Bias and Threshold UE assignment} \label{biasthresh}
      \begin{algorithmic}[1]
        \State \textbf{Run Algorithm \ref{biased}}     
        \While {$\exists m : |\mu(m)|=j$} 
        \If {$\frac{\tilde{\beta}_{k,\alpha} - \tilde{\beta}_{k,\alpha'}}{\tilde{\beta}_{k,\alpha'}} < \delta $} 
            \State $k \gets \mu(m)$ ; $\mathcal{M} \gets \mathcal{M} \setminus m$
            \State $\mu(k) \gets \max\limits_{m(\alpha') \in \mathcal{M}} \tilde{\beta}_{k,\alpha'}$.
        \EndIf
        \EndWhile        
      \end{algorithmic}
    \end{algorithm}
  \end{minipage}
\end{figure}


\subsection{PRB, Spatial Layers and Power Optimization}
\label{subsec:resourceallocation}

Only part of the optimization problem in (\ref{eq:optprob}) is obtained after UE to BS assignment. Allocation of PRBs of BS $m$ to all assigned UEs $k \in \mu(m)$ and power allocation has to be performed, too. It should be noted that the PRB and power allocation is performed for all six UE to BS assignment algorithms proposed in Section \ref{subsubsec:alogs}.   

\subsubsection{Analytical Results}
\label{subsubsec:analyticalresultsallocation}
We make the following assumptions partly based on the current state of the art as well as based on the goal to perform the allocation and control more efficiently. 

At first, all UEs assigned to one BS $m$ obtain the same power allocation, i.e., $p_{k,\alpha} = p_m \; \forall k,\alpha(k)$, because the granularity to adapt at one BS is obtained from the assignment of the number of PRBs. From the results in the preliminary results section (Section \ref{subsec:probstatement}), we know that there exists an optimum number of PRBs $\alpha^*$ assigned to fulfill the rate requirement of the UE $k$. By computing the closed form expression for this number from (\ref{eq: opt_prb_single_user}), we adapt to the large-scale fading and rate requirements of a particular UE.  
Furthermore, we consider only two PRB assignment strategies, namely no spatial multiplexing (NSM) and spatial multiplexing (SM), as illustrated in Figure \ref{fig:PRBassignment}. In NSM, the total number of PRBs available for BS $m$ in band $b$ are optimally distributed among the UEs according to their rate requirement and channel gains. In contrast, via SM each UE obtains a separate spatial layer and can use up to the maximum number of available PRBs, which reduces the spatial degrees of freedom to $(a_m - |\mu(\alpha)|)$.  
\begin{figure}[htp]
    \centering
    \includegraphics[width=0.5\textwidth]{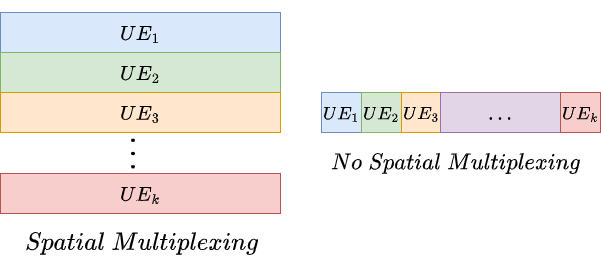}
    \caption{PRB assignment strategies: SM where each UE is assigned a spatial layer and NSM with a single spatial layer.}
    \label{fig:PRBassignment}
\end{figure}

To characterize the proposed PRB allocation and power control strategy, we have the following result, where the sufficient condition for optimal power allocation for NSM and SM is given by  
\begin{eqnarray}
\alpha_1(\tilde{p}_1)\tilde{p}_1 - \alpha_1(p_1^*)p_1^* &\geq& \delta C\eta_{PA}, \; \text{where}, \; \delta = \alpha_1(p_{1}^*) - \alpha_1(\tilde{p}_{1}),
\label{eq: cond_sufficent}
\end{eqnarray}
which states that the difference between load-dependent part with $\tilde{p}$ and $p^*$, i.e., $\alpha_1(\tilde{p}_1)\tilde{p}_1 - \alpha_1(p_1^*)p_1^*$ must be greater than or equal to the load-independent part $\delta C\eta_{PA}$ (Appendix \ref{sec: proof_theo2}), otherwise the load-independent power consumption dominates.

\begin{theorem}
Within a cell applying ZF massive MIMO precoding, for fixed inter-cell interference, if the condition in (\ref{eq: cond_sufficent}) is fulfilled for all $\tilde{p} > p^*$, then, in the case of NSM, the optimal power allocation is obtained by finding the power $p_m^*$ such that the sum of all PRBs assigned to the UEs equals the maximum number of available PRBs $N_m$, i.e., 
\begin{eqnarray}
    \sum_{k \in \mu(m)} \alpha_k^*(p_{m}^*) = N_m \label{eq:cNSM}
\end{eqnarray}
and for SM\begin{change},\end{change} the optimal individual power allocation, is obtained by finding the power $\boldsymbol{p}^* = [p_1^*, \dots, p_k^*]$ such that each UE is assigned the maximum number of available PRBs $N_m$, i.e.,
\begin{eqnarray}
     (\alpha_1^*(p_{1}^*), \dots, \alpha_k^*(p_k^*))_{k \in \mu(m)} = N_m \label{eq:cSM}. 
\end{eqnarray}
\label{theo:2}
\end{theorem}

\begin{proof}
\begin{change}The proof demonstrates by contradiction that deviating from the strategies in (\ref{eq:cNSM}) and (\ref{eq:cSM}) cannot save energy and is provided in Appendix \ref{sec: proof_theo2}\end{change}. The proof shows that any other power allocation $\Tilde{p}_{m}$ does not satisfy (\ref{eq:cNSM}) and (\ref{eq:cSM}) for NSM and SM, respectively, either leading to outages, if $\tilde{p}_{m} < p_{m}^{*}$ or leading to a lower EE (due to an increase in the total power consumption of the BS), if $\tilde{p}_{m} > p_{m}^{*}$.

\end{proof}

\begin{remark}
The sufficient condition in (\ref{eq: cond_sufficent}) depends on the energy consumption parameters. Usually, it is fulfilled, because $D_{m,1}$ is very small and therefore, the costs for using additional PRBs are small compared to the gain in terms of power consumption.      
\end{remark}

\subsubsection{Algorithms}
\label{subsubsec:alogsallocation}
This result directly leads to the Algorithm \ref{NSMpopt}, which is implemented for the assignment of PRBs and power control. It should be noted that the PRB assignment and power control can be executed independently at each cell in a multi-cell network in a decentralized manner \cite{lee1995downlink} because each BS only requires two pieces of local information\begin{change}: the channel gains between itself and its associated UEs and the interference plus noise power at the associated UEs.\end{change} Therefore, each BS can independently perform the PRB allocation and power control decisions for the UEs it serves, taking into account its own local conditions and optimizing its performance based on the available resources and UE requirements, without relying on control from a central entity. By implementing the proposed algorithms in a distributed manner, the computational complexity and the need for extensive inter-cell communication can be avoided. This approach enables scalability and improves the overall efficiency of the multi-cell network. 


\begin{algorithm}
\caption{PRB allocation and power control for NSM and SM}\label{NSMpopt}
\begin{algorithmic}
\State $p_m^0 = p_m$ for all $m \in \mathcal{M}$
\While {$\sum_m |p_{m}^{t+1} - p_{m}^{t}| \neq 0$} 
\For{$m=1..M$}
\If{NSM}
    \State Interference $I_m^t$ at BS $m$ is computed based on the power allocation of the other cells  
    \State $p_m^* \gets p_m^t : \sum_{k \in \mu(m)} \alpha_k^*(p_m^t,I_m^t) = N_m$
\Else
    \State Interference $I_m^t$ at BS $m$ is computed based on the power allocation of the other cells  
    \State $p_m^* \gets p_m^t :  (\alpha_k^*(p_m^t,I_m^t))_{k \in \mu(m)} = N_m$
\EndIf
\EndFor
\EndWhile
\end{algorithmic}
\end{algorithm}

The convergence of the Algorithm \ref{NSMpopt} depends on the feasibility of the rate requirements under the power constraints \cite{8641436}. Before the algorithm is started, we check the feasibility of the requested rates using Perron-Frobenius theory \cite{castaneda2013sum}. If feasible, the iterative algorithm converges when the allocated power per BS or UE does not change between the current and previous iterations, i.e., an optimum power allocation is achieved.

\subsection{Computational Complexity of the Proposed Algorithms}

We make a distinction between the computational complexity of the UE to BS assignment (Algorithms 1-6) and the PRB allocation and power control (Algorithm 7).  

A summary of the computational complexity can be found in Table \ref{tab:complexity}, where $K$ and $M$ represent the number of UEs and the number of BSs, respectively.

\begin{table}[htp]
    \centering
    \begin{adjustbox}{width=0.6\textwidth}
    {\begin{tabular}{|l|l|}
    \hline
        Algorithm & Complexity \\ \hline
        Benchmark & $O(K)$ \\ \hline
        Algorithm \ref{greedy} \& \ref{biased} & $O(K)$ \\ \hline
        Algorithm \ref{rematch} \& \ref{biasrematch} & $O(K(M+1))$ \\ \hline
        Algorithm \ref{threshold} \& \ref{biasthresh} & $O(K(M+1))$ \\ \hline
        Algorithm \ref{NSMpopt} & $T M T_B; T_B = \begin{cases}
               \max(n_1, n_2, n_3) * K  & \quad \text{NSM} \\
                 K           & \quad \text{SM } 
            \end{cases}$ \\ \hline
    \end{tabular}}
    \end{adjustbox}
    \caption{Computational Complexity of Algorithms \ref{greedy}-\ref{NSMpopt}.}
    \label{tab:complexity}

\end{table}

\begin{itemize}
    \item \textbf{UE to BS Assignment:} Firstly, the Greedy UE assignment algorithm scales linearly with the number of UEs, and therefore, has the computational complexity of order $K$, i.e. O($K$). Secondly, the RE-Matching UE assignment consists of two steps, where the first step has the the same complexity as Greedy UE assignment, i.e. $O(K)$, and the second step has the complexity $O(MK)$, where $M$ is the number BSs. Therefore, the overall complexity of RE-Matching UE assignment is $O(K(M+1))$. Similarly, the Threshold UE assignment has the complexity of $O(K(M+1))$. Thirdly, the addition of the bias in the Bias UE assignment has negligible computational complexity, and therefore, the Bias UE assignment has the same complexity as the Greedy UE assignment, i.e., $O(K)$. Finally, following the above, the computational complexity of Bias and RE-Matching and Bias and Threshold UE assignment algorithms have similar complexity as Re-Matching and Threshold UE assignment, respectively, i.e. $O(K(M+1))$.

    \item \textbf{PRB Allocation and Power Control:} Let $T$ be the number of iterations required for the algorithm to converge, i.e. the stopping criteria in the while loop. Then, the computational complexity of Algorithm $7$ is $O(TMT_{B})$, where $M$ is the number of BSs. The computational complexity for $T_B$ is dependent on the multiplexing technique. In the case of NSM, we use bisection search to allocate power to the UEs, such that with optimal power $\sum_{k \in \mu(m)} \alpha_k^*(p_{m}^*) = N_m$ is satisfied. Whereas, for SM with fixed interference and the condition  $(\alpha_1^*(p_{1}^*), \dots, \alpha_k^*(p_k^*))_{k \in \mu(m)} = N_m$ at optimal power, we use the achievable rate expression to compute the allocated power per UE. Therefore, the computational complexity of  $T_B = \max(n_1, n_2, n_3) \cdot K$ for NSM, where $n_1$, $n_2$ and $n_3$ represent the number of iterations needed for the three while loops used for bisection search and $T_B = K$ for SM.

\end{itemize}

From the analysis above, it can be seen that our proposed UE to BS assignment, PRB allocation and power control have low complexity, and can be efficiently implemented in a distributed manner. In addition, to compare our proposed distributed PRB allocation and power control algorithm, we propose a benchmark algorithm, where the maximum available power per BS is equally distributed among the UEs assigned and the UEs are spatially multiplexed. We then compute the number of PRBs required to meet the minimum rate requirement of each UE. Any UE requiring more than the maximum available PRBs is considered as an outage. 

The complexity of the proposed benchmark scheme scales linearly with the number of UEs, i.e. $O(K)$, where the computation of the number of PRBs is negligible. However, the performance of such an algorithm is poor in comparison to our proposed algorithm, as shown in Figure \ref{fig:ee2.6network}. Furthermore, to solve the optimization problem in (9) (see manuscript), any direct search method would require an exhaustive search of all possible PRB assignment, followed by finding the optimal power allocation for each such assignment. For such an approach, it is apparent that the complexity is exponential in the number of BSs, UEs and PRBs\cite{6678362}.

\section{Numerical Results}
\label{sec:numericalresults}
All network simulations \begin{change}were conducted\end{change} using {SiMoNe} \cite{Rose16}, which is designed to simulate complex wireless communication scenarios as realistically as possible.

\subsection{Simulation Setup}
\label{subsec:simusetup}
\begin{change}Varying network structures were implemented, and multiple simulation parameters were considered. The metrics measured for coverage optimization include reference signal received power (RSRP) and SINR. In LTE networks, RSRP is used to measure the UE's received signal strength over the common reference signal, with a fixed threshold of (RSRP $> -115$ dBm) to measure the network's coverage quality. SINR measures the the UE signal quality, and a threshold of (SINR $> -6.5$dB) was set to measure the outage probability.\end{change} 

In addition, the metrics describing the macroscopic scenario simulation include the location, antennas and subscribers/UEs. An outdoor scenario of the size ($7 * 13$) km was created in the city center of \begin{change}Berlin, Germany. 3D building data was taken from the actual city data, and realistic antenna deployment was implemented across the city based on an actual network antenna density distribution.\end{change}  

Macroscopic antennas were deployed within the scenario, \begin{change}where the number of antennas is dependent on their frequencies; the higher the frequency, the more antennas were inserted. Multiple antenna deployment were tested to maintain a threshold of $RSRP_{ue} < - 115$ dBm and ensure an efficient antenna deployment from an energy perspective. The antenna layout was fixed to a $3$-sector antenna per site with azimuth angles of (0, 120, 240)°. Additionally, different mechanical tilts of (2, 4, 6)° for the antenna main lobe were tested following the antenna deployment procedure. The metrics were measured, and the deployment was calibrated based on the results. Furthermore, UEs were deployed homogeneously across the map to simulate outdoor subscribers.\end{change}

\begin{figure}[htp]
     \centering
     \begin{subfigure}[b]{0.47\textwidth}
         \centering
         \fbox{\includegraphics[width=0.88\linewidth,height=5.0cm]{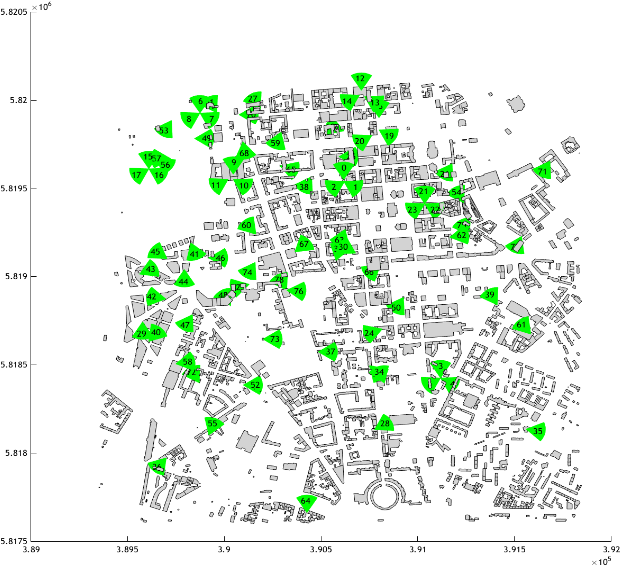}}
         \caption{700 MHz}
         \label{fig:700MHz_AntennaDistribution}
     \end{subfigure}
     \hfill
     \begin{subfigure}[b]{0.49\textwidth}
         \centering
         \fbox{\includegraphics[width=0.9\linewidth,height=5.0cm]{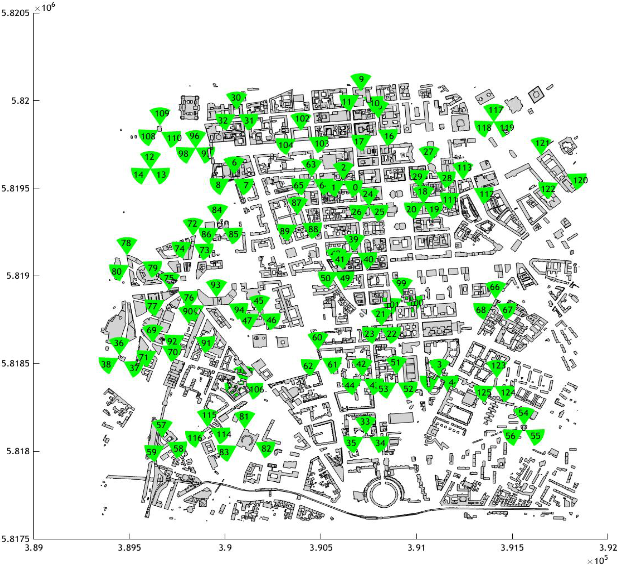}}
         \caption{2.6 GHz}
         \label{fig:2.6GHz_AntennaDistribution}
     \end{subfigure}
     \hfill
        \caption{Antennas distributed based on a realistic antenna distribution in downtown Berlin.}
        \label{fig:Antenna Distribution}
\end{figure}

To obtain sufficient data from the simulation, multiple scenarios with two different frequencies, i.e., $700$ MHz and $2.6$ GHz,  were tested. Table \ref{tab:700MHZ_2.6GHzParameters} shows the evaluation parameters for the Berlin scenario. For $700$ MHz band (see Figure \ref{fig:700MHz_AntennaDistribution}), $13$ outdoor BS were deployed since lower frequencies travel greater distances and provide higher coverage, which is reflected in Figure \ref{fig:700MHZ_RSRP&SINR}. Furthermore, it was observed that a mechanical tilt of $2^{\circ}$ provides the best results. Therefore, the RSRP values at this mechanical tilt were selected. 

\begin{table}[htp]
\centering
    \begin{tabular}{| p{.3\linewidth} | p{.6\linewidth} |  }
    \hline
   \rowcolor{lightgray} \textbf{Parameters} & \textbf{Values} \\ \hline
    Carrier Frequency & 700 MHz, 2.6 GHz  \\ \hline
    Antenna Type   & Directional Antenna   \\ \hline
    Antenna Gain (dBi) & 8 for 700 MHz, 18.6 for 2.6 GHz \\ \hline
    Cell Layout  & 700 MHz - 13 sites, 2.6 GHz - 42 sites; 3 sectors per site   \\ \hline
    Bandwidth (MHz) & 10 for 700 MHz, 100 for 2.6 GHz   \\ \hline
    Antenna Diagram  & 3GPP specified for 700 MHz, Commscope-HWXX-6516DS-VTM2600 for 2.6 GHz   \\ \hline
    Antenna Transmit Power (dBm) & 46 \\ \hline
    Tilt Type & Mechanical Tilt  (2, 4, 6)°  \\ \hline
    Effective Isotropic Radiated Power (dBm) & 63.5 \\ \hline
    BS Antenna Height (m) & Different Heights depending on building data with $\max_h = 59,356 $ and $\min_h = 5 $\\ \hline
    Number of  UEs  & 195   \\ \hline
    UE Transmit Power (dBm) & 23 \\ \hline
    UE Distribution & Homogeneous \\ \hline
    UE Height (m) & 1.5  \\ \hline
    UE Type \& Mobility & Pedestrian - Stationary  \\ \hline  
   
    \end{tabular}
    \caption{700 MHz and 2.6 GHz simulation parameters.}
    \label{tab:700MHZ_2.6GHzParameters}

\end{table}

\begin{figure}[htp]
\begin{subfigure}{.5\textwidth}
  \centering
  \fbox{\includegraphics[width=.8\linewidth, height = 4cm]{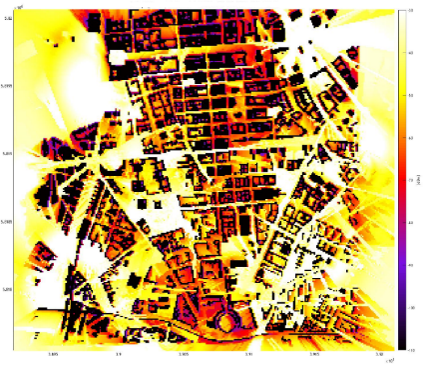}}
  \caption{ Reference Signal Received Power }
  \label{fig:sfig1}
\end{subfigure}%
\begin{subfigure}{.5\textwidth}
  \centering
  \fbox{\includegraphics[width=.8\linewidth, height = 4cm]{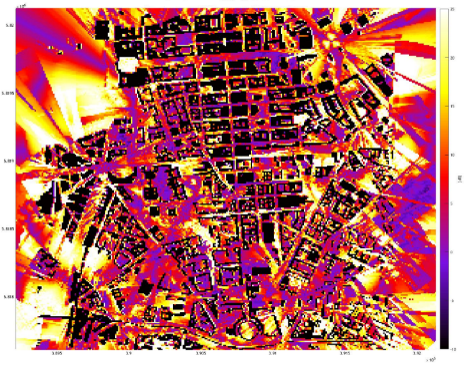}}
  \caption{Signal to Interference and Noise Ratio }
  \label{fig:sfig2}
\end{subfigure}
\caption{Prediction of RSRP and SINR with 13 BS distribution scenario in the 700 MHz frequency.}
\label{fig:700MHZ_RSRP&SINR}
\end{figure}

Similarly, for the 2.6 GHz band, the number of outdoor sector antennas \begin{change}was\end{change} increased to $42$ (see Figure \ref{fig:2.6GHz_AntennaDistribution}) to ensure a higher coverage, which is reflected in Figure \ref{fig:2.6GHZ_RSRP&SINR}. In addition, Figure \ref{fig:2.6GHz_CDF} shows the CDF of the RSRP values at different mechanical tilts, where,  mechanical tilt = $2^{\circ}$ provided the best results and, therefore, the RSRP values were selected based on this tilt.

\begin{figure}[htp]
\begin{subfigure}{.5\textwidth}
  \centering
  \fbox{\includegraphics[width=.8\linewidth, height = 4cm]{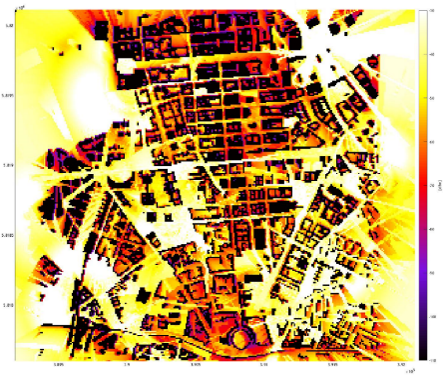}}
  \caption{ Reference Signal Received Power }
  \label{fig:sfig_1}
\end{subfigure}%
\begin{subfigure}{.5\textwidth}
  \centering
  \fbox{\includegraphics[width=.8\linewidth, height = 4cm]{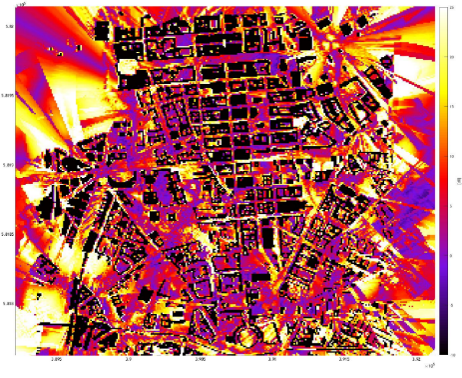}}
  \caption{Signal to Interference \& Noise Ratio }
  \label{fig:sfig_2}
\end{subfigure}
\caption{Prediction of RSRP and SINR with 42 BS distribution scenario in the 2.6 GHz frequency.}
\label{fig:2.6GHZ_RSRP&SINR}
\end{figure}

\begin{figure}[H]
    \centering
    \includegraphics[width=.52\linewidth]{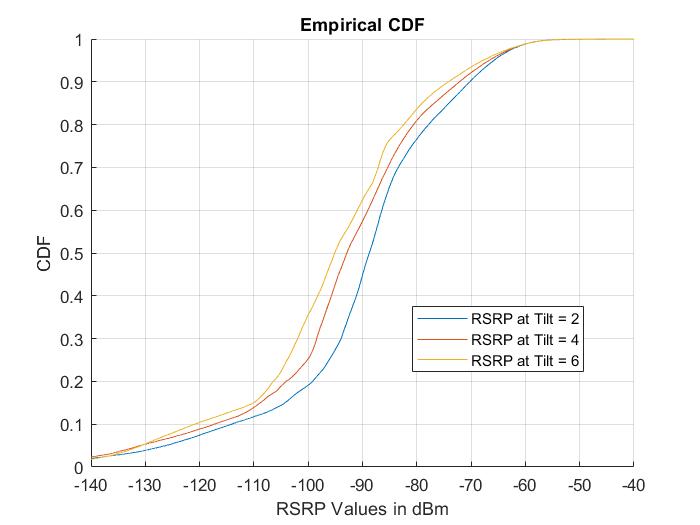}
    \caption{CDF of the 2.6 GHz band at three different mechanical tilts.}
    \label{fig:2.6GHz_CDF}
\end{figure}

\subsection{Convergence of the Proposed Distributed Algorithm}
\label{subsec:algo_convergence}
In Section \ref{subsubsec:alogsallocation}, we mention that the convergence of the proposed distributed Algorithm \ref{NSMpopt} is related to the power constraints. \begin{change}Figure \ref{fig:convergence_plot} verifies this statement, showing the number of iterations required for convergence. The y-axis represents the power consumed by the BS. The two curves can be distinguished based on the initial power assignment, where maximum power initialization represents that the maximum available transmit power is equally split among each assigned UE and, with zero power initialization each UE assigned to the BS is initialized with zero transmit power. \end{change}

We observe that the algorithm converges within $6$ or $7$ iterations in this case, depending on the targeted accuracy level. However, with a slightly relaxed accuracy level, the algorithm can converge in $3$ or $4$ iterations, \begin{change}as\end{change} inferred from the figure.             

\begin{figure}[htp]
\centering
\includegraphics[width=.45\linewidth]{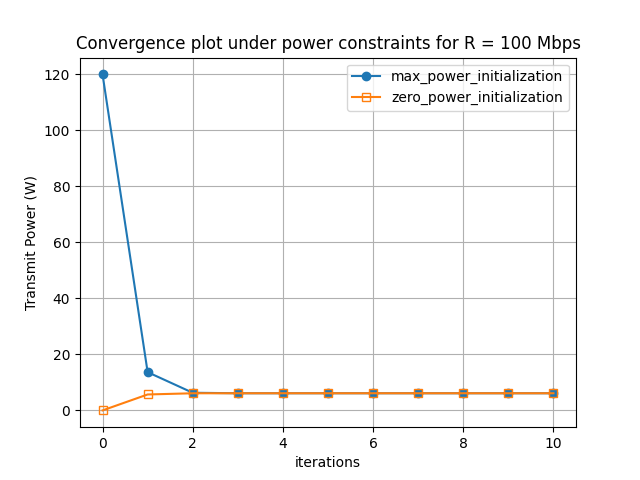}    
\caption{Convergence of the proposed distributed Algorithm 7.}
\label{fig:convergence_plot}
\end{figure}

\subsection{Simulation Results - Individual Bands}
\label{subsec:singlecellres}
In this section, we \begin{change}present\end{change} the simulation results for the setup explained in Section \ref{subsec:simusetup}, considering one active band and $195$ UEs. \begin{change}Additionally\end{change}, we assume that a $100$ Megabit file is to be transferred, and the maximum number of available antennas are used per BS. Furthermore, each UE has the same minimum rate requirement. 

\begin{table}[htp]
    \centering
    \begin{adjustbox}{width=0.6\textwidth}
    \begin{tabular}{|c|c|c|c|c|c|}
    \hline
        $D_0$ & $D_1 - 700$ MHz & $D_1 - 2.6$ GHz & $P_{FIX}$ & $P_{SYNC}$ & $\eta_{PA}$ \\ \hline
        4.49 W & 0.00312 W & 0.01560 W & 300.0 W & 34.0 W & 0.48 \\ \hline
    \end{tabular} 
    \end{adjustbox}
    \caption{Parameters for the computation of energy consumption.}
    \label{tab: sys_para1}
\end{table}

Tables \ref{tab: sys_para1} and \ref{tab: sys_para2} \begin{change}display\end{change} the value of the parameters used throughout the simulations. In Table \ref{tab: sys_para2},  $P^{max}$ \begin{change}represents\end{change} the maximum transmit power per BS,  $M$ \begin{change}denotes\end{change} the number of BSs, and $RL$ and $SL$ are the abbreviations used to represent the location of the BSs. \begin{change}Here $RL$ denotes the refreshed (different) locations, while $SL$ means same locations\footnote{\color{blue}Please note that the locations of BSs in $2.6$ GHz band were fixed, while the locations of BSs in $700$ MHz band were changed. Therefore, the locations of BSs in $700$ MHz are relative to the locations of BSs in $2.6$ GHz band. In this context, 'refreshed location' indicates that BS locations in the 700 MHz band differ from those in the 2.6 GHz band.}.\end{change} 
\begin{table}[htp]
    \centering
    \begin{adjustbox}{width=0.6\textwidth}
    \begin{tabular}{|c|c|c|c|c|c|c|}
    \hline
        ~ & ${A}_m^b$ & $N_{m}^b$  & $P^{max}$ & $\bar{b}$ & M - RL & M - SL \\ \hline
        700 MHz & 4 & 100 & 200 W & 180 KHz & 39 & 126 \\ \hline
        2.6 GHz & 64 & 273 & 120 W & 360 KHz & 126 & 126 \\ \hline
    \end{tabular}
    \end{adjustbox}
    \caption{Parameters highlighting the maximum number of available resources.}
    \label{tab: sys_para2}
\end{table}


\subsubsection{Simulation Results for $700$ MHz Band}
\label{subsubsec:700MHz}
Figure \ref{fig: cdf_ues_700} \begin{change}illustrates\end{change} the distribution of UEs among the BSs using the greedy approach described in Algorithm \ref{greedy}. \begin{change}About $10\%$ of the BSs have no UEs assigned when the BSs in $700$ MHz band are located differently (RL), while more than $60\%$ of the BSs have no UEs assigned when BSs in $700$ MHz band are located at the same positions (SL). This is due to the variation in the number of BSs in the $700$ MHz band in the two scenarios. Figure \ref{fig:outages_700} further illustrates the impact of having fewer BSs. A lower total number of BSs in the network results in more outages due to highly congested BSs, leading to a lower EE.\end{change}

    
\begin{figure}[htp] 
\centering
\includegraphics[width=0.45\linewidth]{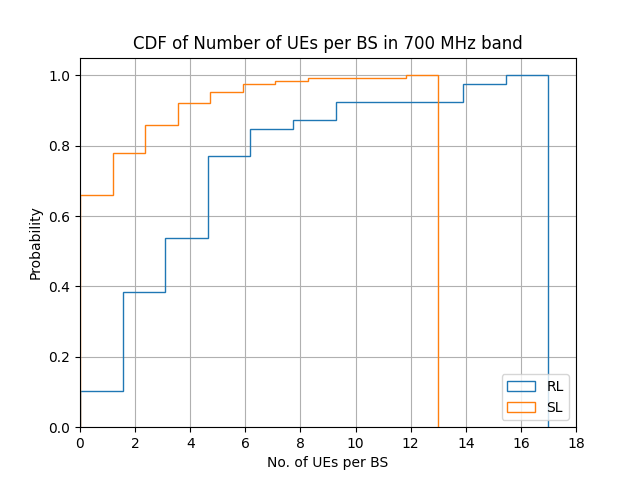}
\caption{CDF of number of active UEs in 700 MHz: RL - different BS location and SL - same BS location relative to BSs in 2.6 GHz band.}
\label{fig: cdf_ues_700}
\end{figure}

\begin{figure}[htp]
\centering
    \begin{subfigure}{0.45\linewidth}
    \centering
    \includegraphics[width=\textwidth]{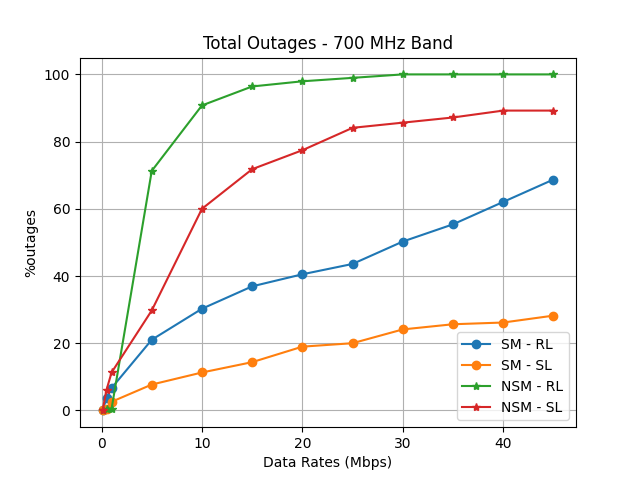}
    \caption{Overall outages.}
    \label{fig:outages_700}
    \end{subfigure}
    \begin{subfigure}{0.45\linewidth}
    \centering
    \includegraphics[width=\textwidth]{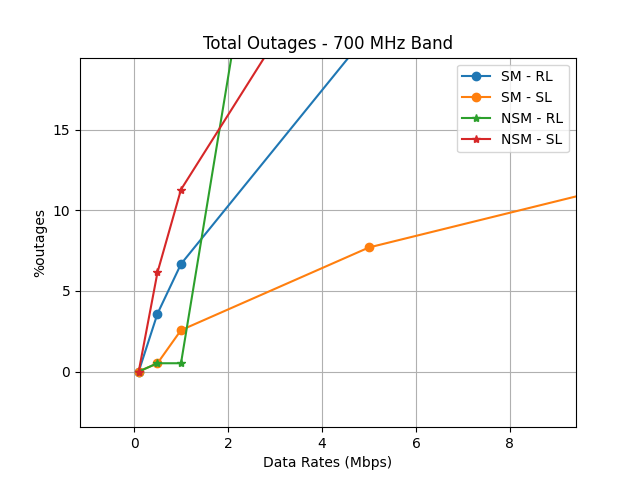}
    \caption{Outages for low data rates.}
    \label{fig:outages_700_zoomed}
    \end{subfigure}
\caption{
Outages over different rate requirements in 700 MHz band: RL - different BS location and SL - same BS location relative to BSs in 2.6 GHz band.}
\label{fig:outages700}
\end{figure} 

\begin{figure}[htp]
\centering
    \begin{subfigure}{0.45\linewidth}
    \centering
    \includegraphics[width=\textwidth]{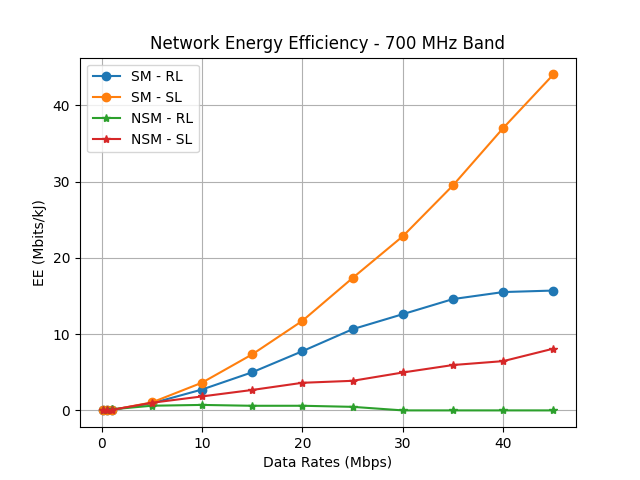}
    \caption{Overall EE of the network.}
    \label{fig:ee_700}
    \end{subfigure}
    \begin{subfigure}{0.45\linewidth}
    \centering
    \includegraphics[width=\textwidth]{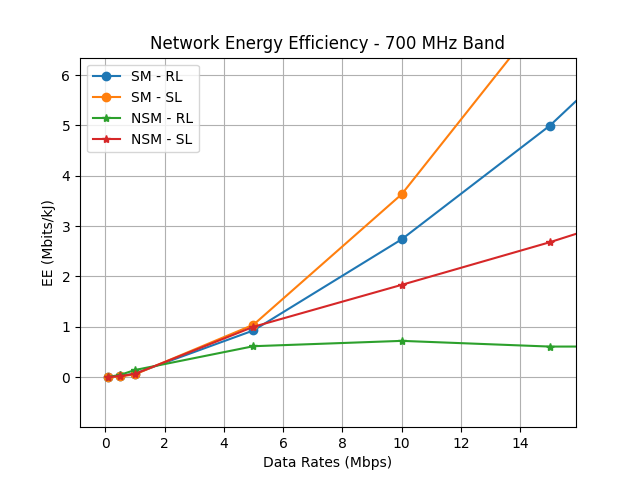}
    \caption{EE of the network for low data rates.}
    \label{fig:ee_700_zoomed}
    \end{subfigure}
\caption{
EE in $700$ MHz band: RL - different BS location and SL - same BS location relative to BSs in 2.6 GHz band.}
\label{fig:ee700}
\end{figure}
Figure \ref{fig:ee_700} shows the EE of the network. \begin{change}Firstly, it is evident that spatially multiplexing UEs achieve higher EE for large data rates, while a single spatial layer results in higher EE for small data rates.\end{change} Such a behavior results from the difference in the number of outages between SM and NSM (see Figure \ref{fig:outages_700}) and also aligns with the conclusions mentioned in \cite{Marwaha2022}. Secondly, low EE is achieved for lower rates, which increases as the rates increase. \begin{change}This results from the longer time required to transmit $100$ Megabits of data for lower data rates, leading to higher power consumption and, consequently, reduced EE.\end{change} Finally, as depicted in Figure \ref{fig:outages_700}, an increase in the number of outages is observed for higher rates. \begin{change}Since the sum rate computation does not consider the impact of outages, the total EE should theoretically be zero in the case of a $100\%$ outage.\end{change} Such a behavior can be observed for NSM (with BSs in RL) when the data rate is greater than $35$ Mbps. Furthermore, for NSM, when the number of UEs to be assigned to BSs is much larger than the number of available BSs and each UE requires a large rate, a higher percentage of outages are observed because the limited available PRBs- to be distributed among the UEs- are insufficient to satisfy the minimum rate requirement. Thus, indicating that higher rates cannot be supported in $700$ MHz frequency band.   

\subsubsection{Simulation Results for $2.6$ GHz Band} \label{subsubsec:2.6GHz} Figure \ref{fig: cdf_ues_2.6} shows the distribution of UEs over BSs and Table \ref{tab: uetobsassignment} explains the meaning of the labels in the legend. It can be seen that more than $60\%$ of the BSs are assigned no UEs in all assignment techniques. Since BSs with only one UE are turned-off in the Rematched 1 assignment, BSs are always assigned two or more UEs. Similarly, three or more UEs are assigned to a BS via Rematched 2 assignment.    
    \begin{table}[htp]
    \centering
    \begin{adjustbox}{width=0.6\textwidth}
    \begin{tabular}{|c|c|}
    \hline
        Baseline & Greedy approach  \\ \hline
        Rematched 1 & BSs with 1 UE switched-off \\ \hline
        Rematched 2 & BSs with 1 and 2 UEs switched-off \\ \hline
        Threshold 1 and 2 & BSs switched-off based on a threshold value \\ \hline
    \end{tabular} 
    \end{adjustbox}
    \caption{Naming of UE to BS assignment techniques.}
    \label{tab: uetobsassignment}
\end{table}

\begin{figure}[htp]
\centering
    \begin{subfigure}{0.45\linewidth}
    \centering
    \includegraphics[width=\textwidth]{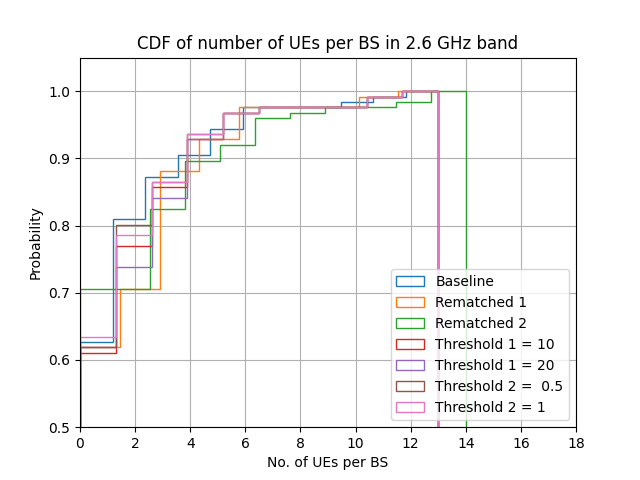}
    \caption{CDF of number of active UEs in 2.6 GHz.}
    \label{fig: cdf_ues_2.6}
    \end{subfigure}
    \begin{subfigure}{0.45\linewidth}
    \centering
    \includegraphics[width=\textwidth]{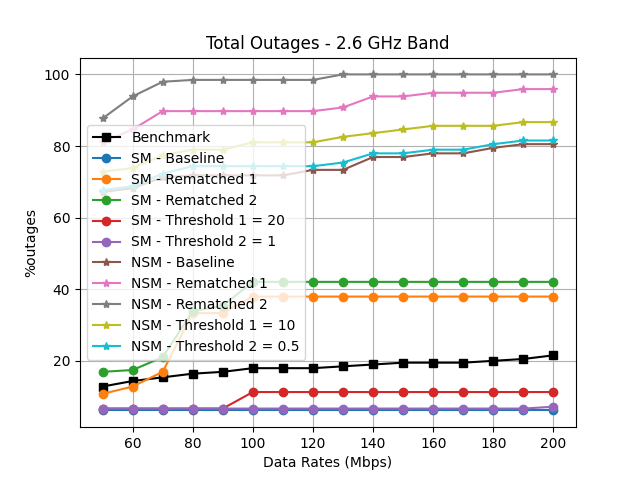}
    \caption{Outages over different rates in 2.6 GHz band.}
    \label{fig: outage2.6}
    \end{subfigure}
\caption{UE distribution over BSs and outages in $2.6$ GHz band.}
\label{fig:cdfoutages2.6}
\end{figure}    

Figure \ref{fig:ee2.6network} shows the changes EE over different rates. Firstly, note that higher rates can be supported in $2.6$ GHz band. Secondly, the EE improves as the rates increase, resulting from less amount of time required to transmit $100$ Mbits of data and lower number of outages (see Figure \ref{fig: outage2.6}). It is apparent that SM significantly outperforms NSM, achieving approximately $3-5$ times larger EE for larger data rates, where, even the benchmark scheme outperforms NSM. Thirdly, among the threshold based UE to BS assignment techniques, it is observed that a higher threshold value, i.e. $20$, results in a higher EE for SM compared to a lower threshold. For NSM, a lower threshold value, i.e. $0.5$, yields a higher EE, thus, indicating that the choice of the threshold depends on the multiplexing scenario. Finally, via the proposed threshold based re-matching algorithm, it is clearly evident that simply switching-off BSs with one and two UEs is not sufficient to yield large EE. A straightforward switch-off of BSs with an arbitrary number of UEs (in these simulations BSs with one and two UEs) can lead to the assignment of UEs to BSs with weak channel gain, thus, resulting in more outages and a lower EE. Therefore, to achieve the highest EE, UEs should be spatially multiplexed and re-matched to the BSs using a carefully selected threshold value. 

\begin{figure}[htp]
\centering
    \begin{subfigure}{0.45\linewidth}
    \centering
    \includegraphics[width=\textwidth]{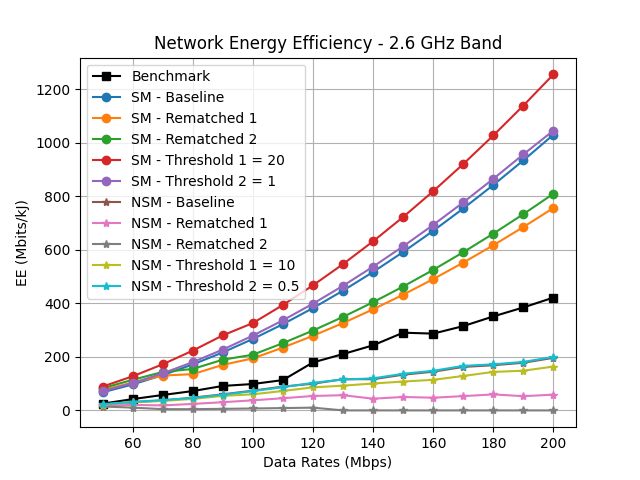}
    \caption{Overall EE of the network.}
    \label{fig:ee2.6}
    \end{subfigure}
    \begin{subfigure}{0.45\linewidth}
    \centering
    \includegraphics[width=\textwidth]{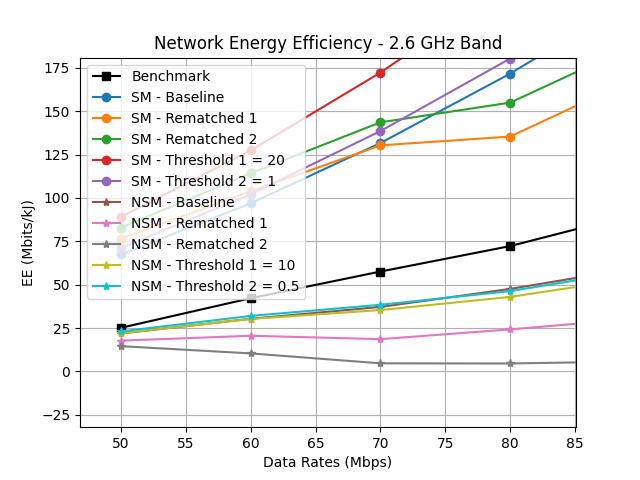}
    \caption{EE of the network for smaller rates.}
    \label{fig:ee2.6zoomed}
    \end{subfigure}
\caption{Total EE of the Network in $2.6$ GHz band.}
\label{fig:ee2.6network}
\end{figure}

Figures \ref{fig:ee2.6ant} and \ref{fig:out} show the impact of the number of active antennas/RF chains on the EE and the percentage of outages of the network, respectively, when the UEs are spatially multiplexed. It can be observed that the number of active antennas/RF chains affects both the EE and the percentage of outages. A lower number of active antennas/RF chains lead to a lower EE because the number of outages are large. As the number of active antennas/RF chains increases, the number of outages decrease and the EE improves, where proportionality between number of active antennas/RF chains and number of outages is observed. However, if the number of active antennas/RF chains grows really large (exemplified by the case with $64$ antennas) the load-independent power consumption begins to dominate, resulting in a lower EE. Furthermore, it should be noted that for each active antenna/RF chain value in Figutre \ref{fig:ee2.6ant}, the UEs are assigned all available PRBs. The results from Figure \ref{fig:ee2.6ant} suggest that from an EE perspective it is better to assign all PRBs before turning on additional antennas/RF chains.           

\begin{figure}[htp]
\centering
    \begin{subfigure}{0.4\linewidth}
    \centering
    \includegraphics[width=\textwidth]{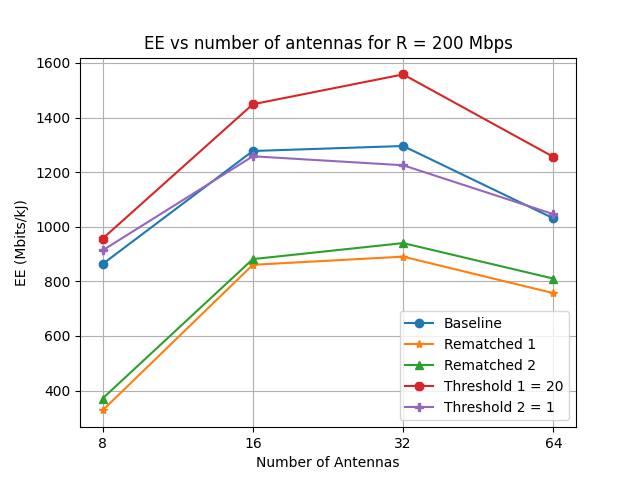}
    \caption{Overall EE of the network.}
    \label{fig:ee2.6ant}
    \end{subfigure}
    \begin{subfigure}{0.4\linewidth}
    \centering
    \includegraphics[width=\textwidth]{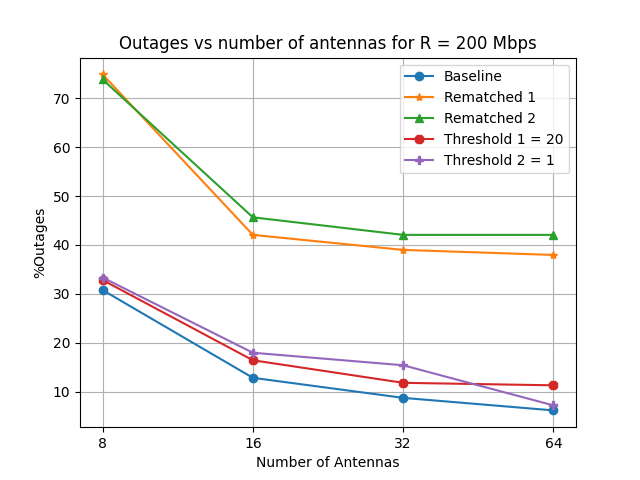}
    \caption{Overall outages of the network.}
    \label{fig:out}
    \end{subfigure}
\caption{Total EE and outages of the network in $2.6$ GHz band.}
\label{fig:ee_ants}
\end{figure}

\subsection{HetNet Simulation Results - Multiple Bands}
\label{subsec:hetnetres}
Shown below are the simulation results for the setup explained in Section \ref{subsec:simusetup} for two active bands with a total of $195$ UEs. We assume that a $100$ Megabits of data is to be transferred, maximum number of available antennas are used per BS and each UE has the same rate requirement. Furthermore, Tables \ref{tab: sys_para1} and \ref{tab: sys_para2} show the system parameters used for the simulations. 

\begin{figure}[htp]
    \centering
        \begin{subfigure}{0.43\linewidth}
        \centering
        \includegraphics[width=\textwidth]{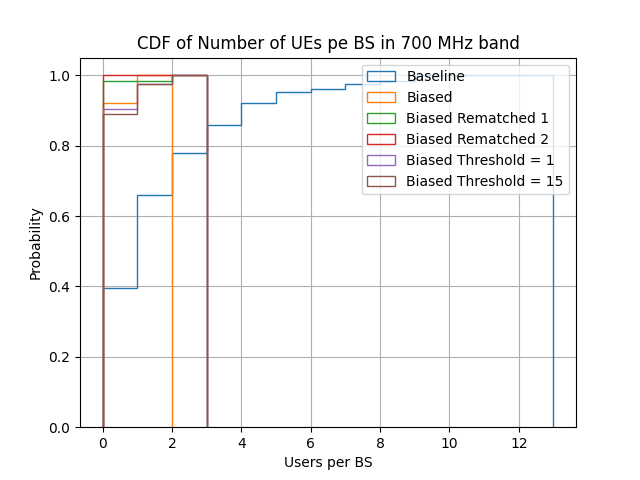}
        \caption{CDF of number of active UEs in $700$ MHz band.}
        \label{fig:cdf700mb}
        \end{subfigure}
        \begin{subfigure}{0.43\linewidth}
        \centering
        \includegraphics[width=\textwidth]{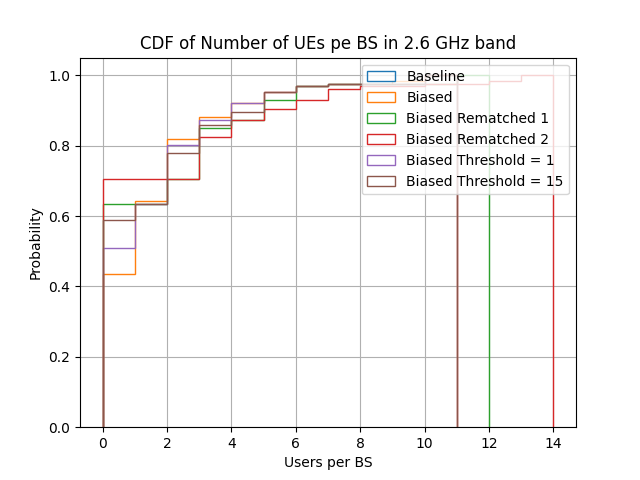}
        \caption{CDF of number of active UEs in $2.6$ GHz band.}
        \label{fig:cdf2.6mb}
        \end{subfigure}
    \caption{Distribution of UEs in a HetNet scenario for different UE to BS matching techniques.}
    \label{fig:cdfsmb}
\end{figure}

\begin{figure}[htp]
    \centering
        \begin{subfigure}{0.43\linewidth}
        \centering
        \includegraphics[width=\textwidth]{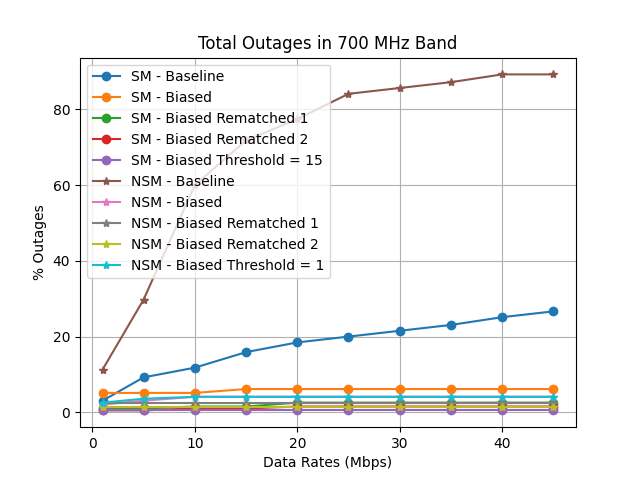}
        \caption{Outages in $700$ MHz band.}
        \label{fig:out700mb}
        \end{subfigure}
        \begin{subfigure}{0.43\linewidth}
        \centering
        \includegraphics[width=\textwidth]{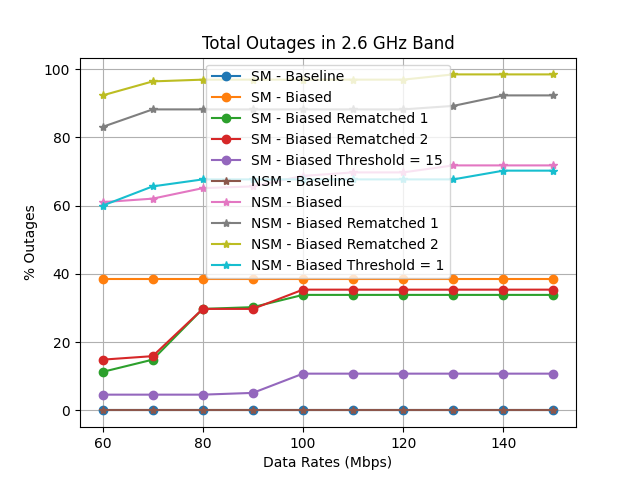}
        \caption{Outages in $2.6$ GHz band.}
        \label{fig:out2.6mb}
        \end{subfigure}
    \caption{Outages in a HetNet scenario for different data rates.}
    \label{fig:outsmb}
\end{figure}

Figures \ref{fig:cdf700mb} and \ref{fig:cdf2.6mb} show the distribution of the UEs over the BSs based on the UE to BS assignment algorithms described in Section \ref{subsec:uebsassignment}. It can be observed that the (greedy) UE to BS assignment rule in (\ref{eq:bua}) leads to a sparsely filled $2.6$ GHz frequency band and a crowded $700$ MHz frequency band. However, as discussed in Section \ref{subsec:singlecellres} and depicted in Figures \ref{fig:out700mb} and \ref{fig:out2.6mb}, a crowded $700$ MHz band leads to higher outages resulting in lower EE. In addition, it is infeasible to assign more than three UEs per PRB when the UEs are spatially multiplexed\footnote{For SM, zero-forcing cannot be applied if the BSs are assigned more than $3$ UEs per PRB due to limited number of antennas.} and high rates cannot be supported. Therefore, a bias of $35$ dB is injected to improve the UE to BS assignment\footnote{However, note that the PRB assignment and power control is performed on the true large-scale fading values.} such that we obtain a sparsely filled $700$ MHz band and a crowded $2.6$ GHz band (Figures \ref{fig:cdf700mb} and \ref{fig:cdf2.6mb}) via the distribution of UEs over BSs and through Table \ref{tab:numues} via the number of UEs assigned in each band.

\begin{table}[htp]
    \centering
    \begin{adjustbox}{width=0.7\textwidth}
    \begin{tabular}{|c|c|c|}
    \hline
        UE Assignment Algorithm  & No. of UEs in 700 MHz  & No. of UEs in 2.6 GHz  \\ \hline
        Basline (Greedy) & 195 & 0 \\ \hline
        Biased & 12 & 183 \\ \hline
        Biased Rematched 1 & 5 & 190 \\ \hline
        Biased Rematched 2 & 3 & 192 \\ \hline
        Biased Threshold (val = 1) & 16 & 179 \\ \hline
        Biased Threshold (val = 15) & 18 & 177 \\ \hline
    \end{tabular}
    \end{adjustbox}
    \caption{Total number of UEs assigned to BSs in $700$ MHz and $2.6$ GHz bands.}
    \label{tab:numues}
\end{table}

\begin{figure}[htp]
    \centering
        \begin{subfigure}{0.55\linewidth}
        \centering
        \includegraphics[width=\textwidth]{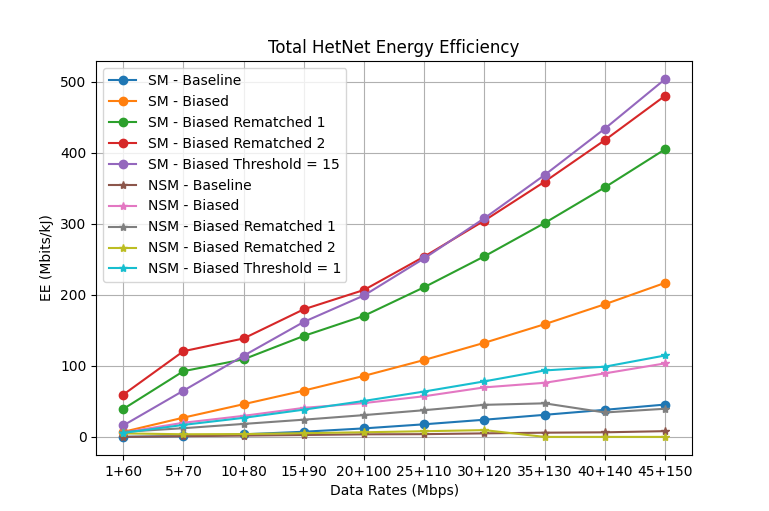}
        \caption{Overall EE of the network.}
        \label{fig:ee7002600}
        \end{subfigure}
        \begin{subfigure}{0.4\linewidth}
        \centering
        \includegraphics[width=\textwidth]{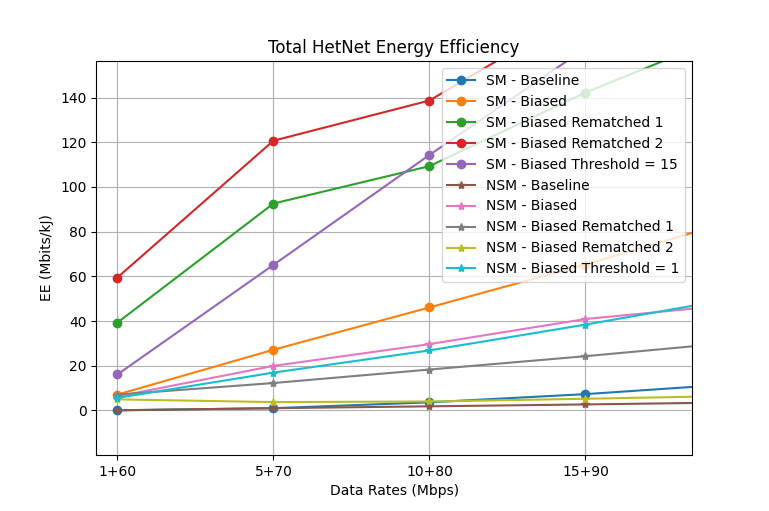}
        \caption{EE of the network for smaller rates.}
        \label{fig:ee7002600zoomed}
        \end{subfigure}
    \caption{Total EE of the HetNet for different UE to BS assignment algorithms.}
    \label{fig:ee7002600network}
    \end{figure}

Figure \ref{fig:ee7002600network} shows the EE of the HetNet network\footnote{Note that the rates shown on the x-axis have the  following format: (Rate in $700$ MHz band + Rate in $2.6$ GHz band). For example, all UEs in $700$ MHz band having a rate of $1$ Mbps and all UEs in $2.6$ GHz band having a rate of $10$ Mbps is represented by $(1 + 10)$ Mbps.}. It can be observed that SM significantly outperforms NSM. In addition, off-loading the UEs to \begin{change}a\end{change} higher frequency band by adding a bias yields much larger EE gains relative to the greedy baseline matching. Furthermore, for SM, switching-off BSs help lower the power consumption, thereby, increasing the EE of the network. Among the re-matching techniques, for lower rates, a simple switch-off of BSs with low EE (in these simulation BSs with one and two UEs were switched-off) provides a higher EE, whereas, for larger data rates, a \begin{change}threshold-based\end{change} reassignment of the UEs yields a higher EE. In contrast, a straightforward switch-off of BSs for NSM leads to much larger outages (Figure \ref{fig:outsmb}) and lower EE because the available PRBs- distributed among the UEs- are insufficient to support the rate requirements. Overall, to achieve high EE the UEs should be spatially multiplexed and each UE should be assigned maximum available PRBs. Depending on the required data rate, either BSs with one and two UEs assigned should be switched-off or a \begin{change}threshold-based\end{change} reassignment of the UEs to BSs should be adopted. Furthermore, lower frequency band should be sparsely filled, while higher frequency bands should be assigned more UEs.

\section{Conclusion and Future Work}
\label{sec:conclusion}
In this work, we derive novel algorithms for energy efficient UE to BS assignment and propose a re-matching algorithm to switch-off BSs with low EE for massive MIMO HetNets. In addition, we propose a PRB assignment and power control algorithm, which can be implemented in a distributed way in a multi-cell network. Our analysis and simulations suggest the following: \emph{i)} $700$ MHz band alone should only be used for very low data rates (\begin{change}for example\end{change} control signaling, RRM). It cannot support higher rates even for a small number of UEs. The low frequency band, without having a higher frequency band, is inadequate to support large data rates and number of UEs, \emph{ii)} $2.6$ GHz band alone should be used to support UEs with large rates, where the UEs should be spatially multiplexed, each UE should be assigned \begin{change}the\end{change} maximum available PRBs and a \begin{change}threshold-based\end{change} reassignment of the UEs to BSs should be adopted to achieve higher EE. In addition, a dense deployment of BSs is \begin{change}a\end{change} necessary condition to enable the advantages of the $2.6$ GHz described above, and \emph{iii)} the combined $700$ MHz and $2.6$ GHz band performance is worse than $2.6$ GHz band alone, because UEs assigned to \begin{change}the $700$ MHz band\end{change} are not flexible and achieve lower rates. Overall, it is recommended that if a BS has multiple antennas, spatial multiplexing should be always used and each UE should be assigned \begin{change}the\end{change} maximum available PRBs (also for low load and \begin{change}a\end{change} low number of UEs).

For future work, numerical assessments with larger number of UEs and higher loads should be performed. We \begin{change}anticipate\end{change} that NSM might be insufficient to support a large number of UEs \begin{change}given\end{change} its limited resources. Furthermore, the proposed re-matching and bias UE assignments \begin{change}might need to be reevaluated to accommodate a larger variety\end{change} of data rates. \begin{change}While considering the bias UE assignment, it's possible to formulate a large-scale fading bias model to optimize the large-scale fading bias value per BS, as demonstrated in \cite{6166483}.\end{change} \begin{change}Additionally\end{change}, the number of active antennas per BS can also be optimized. 

\begin{change}Another intriguing avenue for future research is understanding the impact of the latest developments in BS hardware on performance and investigating the optimality of the proposed algorithms under these new settings. This exploration can be conducted both numerically and analytically based on the derived rate and energy consumption expressions.\end{change}

\appendices

\section {Proof Justifying the replacement of $\underline{R}_k$ with $R_k$ in \eqref{eq:EEm}}
\label{sec: replacement_proof}
In the following, we justify the replacement of $\underline{R}_k$ with $R_k$ by analysing the KKT conditions for a single cell and extending the results to the multi-cell scenario. 

We begin by analysing the total power consumption (P) for a single cell, which with equal power allocated across all the PRBs assigned to a UE associated with this cell, can be written as  

\begin{eqnarray}
    P = \frac{1}{\eta_{PA}}\sum_{k}N_k \, p_{k} +  I\left(\frac{1}{\lambda_{0}}P_{FIX} + \frac{1}{\lambda_{1}}P_{SYNC} \right) + D_{0} a + \left(D_{1} a \left(\sum_{k}N_k \right) \right) 
    + C. 
 \label{eq:pcscell}
\end{eqnarray}

Note that only the first term in (\ref{eq:pcscell}) depends on the transmit power allocated to the UEs in this cell. Therefore, for simplicity, we drop the remaining terms. Under the assumption that the UE to BS assignment has been fixed, the total power minimization can be defined as

\begin{subequations}    
\begin{alignat}{2}
\min_{\boldsymbol{p} \geq 0}   &\quad&     &  \sum_{k=1}^{K} N_k p_{k} \label{eq:objective}\\
\text{s.t.}                    &\quad&     & \sum_{k=1}^{K} p_{k} \leq P^{\max}  \label{eq:powerbudget} \\
                               &\quad&     & (\Bar{b}\cdot N_k) \log_2 (1 + (M-K) p_{k} \Tilde{\beta_{k}}) \geq \underline{R}_{k} \; , \forall k=1,\dots, K, \label{eq:ratecons}
\end{alignat}
\label{eq:optproblem}
\end{subequations}
where $N_K$ is the number of PRBs assigned to a UE, $\Tilde{\beta_{k}}$ represents the large scale fading normalized by the noise power, $M$ is the number of antennas, and $K$ is the number of multiplexed UEs. The minimum rate constraint (\ref{eq:ratecons}) can be rewritten as $p_{k} \geq \left(\frac{2^{\frac{\underline{R}_{k}}{(\Bar{b}\cdot N_k)}} - 1}{\Tilde{\beta_{k}}(M-K)}\right), \forall k=1,\dots,K$. 
Let us denote $P^{\min}_{k}=\left(\frac{2^{\frac{\underline{R}_{k}}{(\Bar{b}\cdot N_k)}} - 1}{\Tilde{\beta_{k}}(M-K)}\right),~\forall k=1,\dots,K$. Then, Problem (\ref{eq:optproblem}) can be reformulated as

\begin{subequations}    
\begin{alignat}{2}
\min_{\boldsymbol{p} \geq 0}   &\quad&     &  \sum_{k=1}^{K} N_k p_{k} \label{eq:obj1}\\
\text{s.t.}                    &\quad&     & \sum_{k=1}^{K} p_{k} \leq P^{\max}  \label{eq:pb1} \\
                               &\quad&     & p_{k}  \geq P^{\min}_{k}, \forall k=1,\dots,K \label{eq:rc1}
\end{alignat}
\label{eq:optprob1}
\end{subequations}

Problem (\ref{eq:optprob1}) is convex in $\boldsymbol{p}$ with an affine feasible set. It can be shown by analyzing the KKT conditions that in the power minimization problem, the allocated power to each UE is assigned to maintain its minimum rate demand. The Slater's condition holds in \eqref{eq:optprob1} since it is convex and there exists $\boldsymbol{p} \geq 0$ satisfying \eqref{eq:pb1} and \eqref{eq:rc1} with strict inequalities. Therefore, the strong duality in \eqref{eq:optprob1} holds. Hence, the KKT conditions are satisfied and the optimal solution $\boldsymbol{p}^*$ can be obtained using the Lagrange dual method \cite{0521833787}. The Lagrange function (lower-bound) of \eqref{eq:optprob1} is given by
\begin{equation}
    L(\boldsymbol{p},\boldsymbol{\mu},\boldsymbol{\delta},\nu) = \sum\limits_{k=1}^{K} N_k p_k 
+ \sum\limits_{k=1}^{K} \mu_k \left( P^{\min}_{k} - p_{k} \right)
+ \sum\limits_{k=1}^{K} \delta_k (-p_k) +\nu \left( \sum\limits_{i=k}^{K} p_k - P^{\text{max}} \right),
\end{equation}
where $\boldsymbol{\mu}=[\mu_1,\dots,\mu_K]$,  $\boldsymbol{\delta}=[\delta_1,\dots,\delta_K]$, and $\nu$ are the Lagrangian multipliers corresponding to the constraints \eqref{eq:rc1}, \eqref{eq:pb1}, and $p_k\geq 0,~k=1,\dots,K$, respectively. The Lagrange dual problem is given by
\begin{subequations}
	\begin{align*}
	\max_{\boldsymbol{\mu},\boldsymbol{\delta},\nu}\hspace{.0 cm}	
	~~ & \inf\limits_{\boldsymbol{p}} \left\{ L(\boldsymbol{p},\boldsymbol{\mu},\boldsymbol{\delta},\nu) \right\}
	\\
	\text{s.t.}~~& \mu_k \geq 0,~\forall k=1,\dots,K,
	\\
	& \delta_k \geq 0,~\forall i=k,\dots,K.
	\end{align*}
\end{subequations}

The KKT conditions are listed below.
\begin{enumerate}
	\item Feasibility of the primal problem \eqref{eq:optprob1}:
	$$
	\textbf{C-1.1:}~p^*_{k}  \geq P^{\min}_{k} \; , \forall k,~~~~~\textbf{C-1.2:}~p^*_k \geq 0,~\forall k,~~~~~\textbf{C-1.3:}~\sum\limits_{k=1}^{K} p^*_k \leq P^{\text{max}}.
	$$
	\item Feasibility of the dual problem:
	$$
	\textbf{C-2.1:}~\mu^*_k \geq 0,~\forall k=1,\dots,K,~~~~~~\textbf{C-2.2:}~\delta^*_k \geq 0,~\forall k=1,\dots,K,~~~~~~\textbf{C-2.3:}~\nu^* \geq 0.
	$$
	\item The complementary slackness conditions:
	$$
	\textbf{C-3.1:}~\mu^*_k \left(  P^{\min}_{k} - p^*_{k} \right) = 0,\forall i=1,\dots,K,
	$$
	$$
	\textbf{C-3.2:}~\delta^*_k p^*_k = 0,~\forall i=1,\dots,K,~~~~~~~~~
	\textbf{C-3.3:}~\nu^*\left( \sum\limits_{i=k}^{K} p^*_k - P^{\text{max}} \right) = 0.
	$$
	\item The condition $\nabla_{\boldsymbol{p}^*} L(\boldsymbol{p}^*,\boldsymbol{\mu}^*,\boldsymbol{\delta}^*,\nu^*) =0$, which implies that
	$$
	\textbf{C-4:}~\frac{\partial L}{\partial p^*_k} = N_k - 
        \mu_k^* + \delta^*_k + \nu^* = 0, ~\forall k=1,\dots,K.
	$$
\end{enumerate}

The primal dual $\delta^*_k,~\forall k=1,\dots,K,$ acts as a slack variable in \textbf{C-4} (due to the KKT condition \textbf{C-2.2}), so it can be eliminated by reformulating the KKT conditions (\textbf{C-4},\textbf{C-2.2}) and \textbf{C-3.2} respectively as
\begin{equation}\label{transf minpow1 KKT}
\nu^* \geq  \mu_k^* - N_k,~\forall k=1,\dots,K,
\end{equation}
and
\begin{equation}\label{transf minpow2 KKT}
p^*_k \left( \nu^* - \left(\mu_k^* - N_k\right) \right) = 0,~\forall k=1,\dots,K.
\end{equation}


If problem \eqref{eq:optprob1} is feasible with $\sum_{k=1}^{K} p_{k} < P^{\max}$, then the optimal solution satisfies $\sum_{k=1}^{K} p^*_{k} < P^{\max}$, meaning that $\nu^*=0$, due to condition \textbf{C-3.3}. In this case, $\mu_k^* - N_k = 0,~\forall k$, meaning that $\mu_k^* = N_k,~\forall k$, so $\mu_k^*>0,~\forall k$ as $N_k > 0$. According to condition \textbf{C-3.1}, we conclude that $P^{\min}_{k} - p^*_{k} = 0,~\forall k$, thus $p^*_{k} = P^{\min}_{k} ~\forall k$, which is equivalent to  $(\Bar{b}\cdot N_k) \log_2 (1 + (M-K) p_{k} \Tilde{\beta_{k}}) = \underline{R}_{k}, ~\forall k$. Therefore, replacing $\underline{R}_k$ with $R_k$ is justifiable. In the case $\underline{R}_k = 0$, no transmit power is allocated to the UE, i.e. $p_k^* = 0$. 

The analysis above can be easily extended to the multi-cell scenario considering fixed inter-cell interference in each iteration of the Algorithm 7, where 
$P^{\min}_{k} = \left(\frac{(2^{\frac{\underline{R}_{k}}{(\Bar{b} N_k)}} - 1) (I_k + \sigma_k)}{{\beta_{k}}(M-K)}\right),\forall k$.

\section{Proof of Theorem~\ref{theo:1}}
\label{sec: proof_theo1}
We illustrate the proof via the considered setup in Figure \ref{fig:interference}, where BS $m_1$ serves few users $k_1, \dots, k_3$, whereas, BS $m_2$ serves users $k_4, \dots, k_6$. We aim at showing that switching-off BS $m_1$ and assigning users $k_1, \dots, k_3$ to BS $m_2$, given that $k_1, \dots, k_3$ has the next best channel with BS $m_2$, is energy optimal under certain conditions (see (\ref{eq:cond_proof}) and (\ref{eq:cond_proof_1})) and the assumption of homogeneous BS with same load-independent energy consumption model\footnote{Note that the proof can be easily extended to heterogeneous BSs with different load-independent energy consumption model.}, and with same PRB allocation and power control strategy.   
\begin{figure}[htp]
    \centering
    \includegraphics[width=.6\linewidth]{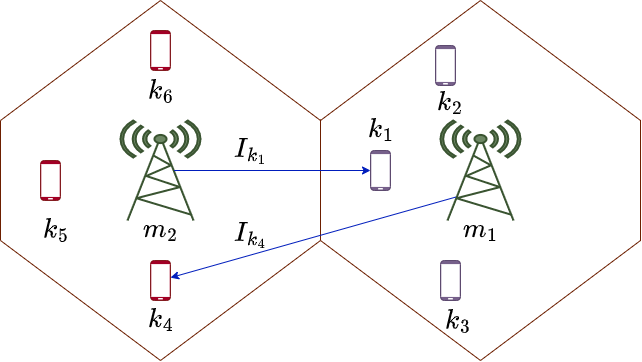}
    \caption{Example scenario to illustrate the proof of Theorem \ref{theo:1}.}
    \label{fig:interference}
\end{figure}

The proof is shown separately for SM and NSM as follows:
\begin{proof}

\emph{i)} \textbf{Spatial Multiplexing}: As depicted in Figure \ref{fig:interference}, consider BS $m_1$ that has only a few users $\{k_1,...,k_\ell\}$ assigned, facing interferences $I_{k_1}, ..., I_{k_\ell}$ from BS $m_2$. For the given minimum rate requirements $\underline{R}_{k_1},...,\underline{R}_{k_\ell}$, the required transmit power for user $k_l$ with PRBs $\alpha_l$ in BS $m_1$ can be computed as 
\begin{equation}
    p_{k_l, \alpha_l} = \frac{(2^{\underline{R'_{k_l}}} - 1)}{a - \ell} \frac{(I_{k_l} + N)}{\beta_{k_l, \alpha_l}},
\end{equation}
for $1 \leq l \leq \ell$ with $\underline{R'_{k_l}} = \frac{\underline{R_{k_l}}}{\bar{b} \cdot N_{m}}$ and maximum available PRBs $N_m$. The total power consumption of BS $m_1$, including the load independent power ($P_{LI}$), is calculated as
\begin{equation}
    P_{m_1} = \frac{N_m}{\eta_{PA}} \sum_{l=1}^\ell p_{k_l, \alpha_l} + P_{LI}. 
\end{equation}

Denote $m_2(l)$ as the second best channel $\beta_{k_l, \alpha'_l}$ for user $k_l$. If all users $k_1,...,k_\ell$ are assigned to their second best BSs $m_2(1),...,m_2(\ell)$ then $P_{m_1}$ amount of power can be saved. However, to support user $k_1,...,k_\ell$ by BS $m_2(1),...,m_2(\ell)$, additional power is needed because  \emph{a)} users $k_1,...,k_
\ell$ needs to be served and \emph{b)} users served by $m_2(1),...,m_2(\ell)$ need more power.

Let us focus on the re-matching of user $l$. Assuming $K_2(l) > 0$ users, $\mu(m_2(l)) = \{k_{l1}, \dots, k_{lK_2}\}$, are already assigned to $m_2(l)$, when user $k_l$ joins BS $m_2(l)$ with PRBs $\alpha'_l$ assigned, the required transmit power can be computed as, 
\begin{equation}
    p_{k_l, \alpha'_l} = \frac{(2^{\underline{R'_l}} - 1)}{a - K_2(l) -1} \frac{(\tilde{I}_{k_l} + N)}{\beta_{k_l, \alpha'_l}}
\end{equation}
where, it is observed that $\tilde{I}_{k_l} < I_{k_l}$ and $\beta_{k_l, \alpha'_l} \leq \beta_{k_l, \alpha_l}$, and the difference between the transmit powers before and after user $k_l$ joins BS $m_2(l)$ can be computed as
\begin{equation}
    \Delta_{k_{2l}} = p'_{k_{2l}, \alpha'_{2l}} - p_{k_{2l}, \alpha'_{2l}}, 
\end{equation}
where $p_{k_{2l}, \alpha'_{2l}}$ is the required transmit power before user $k_l$ joins BS $m_2(l)$ and is computed as 
\begin{equation}
    p_{k_{2l}, \alpha'_{2l}} = \frac{(2^{\underline{R'_{k_{2l}}}} - 1)}{a - K_2(l)} \frac{(I_{k_{2l}} + N)}{\beta_{k_{2l}}, \alpha'_{2l}}, 
\end{equation}
and $p'_{k_{2l}, \alpha'_{2l}}$ is the required transmit power after user $k_l$ joins BS $m_2(l)$ and is computed as
\begin{equation}
    p'_{k_{2l}, \alpha'_{2l}} = \frac{(2^{\underline{R'_{k_{2l}}}} - 1)}{a - K_2(l) - 1} \frac{(\tilde{I}_{k_{2l}} + N)}{\beta_{k_{2l}, \alpha'_{2l}}}.
\end{equation}

Then, 
\begin{equation}
    \Delta_{k_{2l}} = \frac{(2^{\underline{R'_{k_{2l}}}} - 1)}{\beta_{k_{2l}, \alpha'_{2l}}} \left(\frac{\tilde{I}_{k_{2l}} + N}{a - K_2(l) - 1} - \frac{I_{k_{2l}} + N}{a - K_2(l)}\right), 
\end{equation}
and the total difference in power accounting for all the users assigned to BS $m_2(l)$ can be calculated as 
\begin{equation}
    \sum_{k \in \mu(m_2(l))} \Delta_k + \Delta_{k_l}
\end{equation}
with
\begin{equation}
    \Delta_{k_l} = (2^{\underline{R'_{k_l}}} - 1) \left(\frac{\tilde{I}_{k_{l}} + N}{\beta_{k_{l}, \alpha'_l}} \frac{1}{a - K_2(l) - 1} - \frac{I_{k_{l}} + N}{\beta_{k_{l}, \alpha_l}} \frac{1}{a - \ell}\right). 
\end{equation}

Assuming the load independent power consumption for all involved BSs is the same, the difference of power consumption before and after re-matching can be obtained and compared to zero as 
\begin{alignat}{2}
  && \sum_{l=1}^\ell \sum_{k \in \mu(m_2(l))} \Delta_k + \Delta_{k_l} - P_{LI} &\leq 0,
  \\ 
  &\implies
  &\sum_{l=1}^\ell \Delta_{k_2(l)} + (2^{\underline{R'_{k_l}}} - 1) \left(\frac{\tilde{I}_{k_{l}} + N}{\beta_{k_{l}, \alpha'_l}} \frac{1}{a - K_2(l) - 1} - \frac{I_{k_{l}} + N}{\beta_{k_{l}, \alpha_l}} \frac{1}{a - \ell}\right) - P_{LI}
  & \leq 0 ,\\  
  &\implies
  &\sum_{l=1}^\ell \frac{\tilde{I}_{k_{l}} + N}{a - K_2(l) - 1} \beta_{k_{l}, \alpha'_l}^{-1}  - \frac{I_{k_{l}} + N}{a - \ell} \beta_{k_{l}, \alpha_l}^{-1}  
  &\leq \sum_{l=1}^\ell \frac{-\Delta_{k_2(l)}}{2^{\underline{R'_{k_l}}} - 1} + P_{LI}, \\
 &\implies
 & \sum_{l=1}^\ell \frac{(a-\ell)(\tilde{I}_{k_{l}} + N)\beta_{k_{l}, \alpha'_l}^{-1} - (a - K_2(l) - 1) (I_{k_{l}} + N)\beta_{k_{l}, \alpha_l}^{-1}}{(a - K_2(l) - 1)(a - \ell)} 
 & \leq \sum_{l=1}^\ell \frac{-\Delta_{k_2(l)}}{2^{\underline{R'_{k_l}}} - 1} + P_{LI}, \label{eq:lastiq} 
\end{alignat}

where $\Delta_{k_2(l)} = \sum_{k \in \mu(m_2(l))} \Delta_k$. Let $\gamma_{1l} = (a-K_2(l)-1)(I_{k_l}+N)$,  $\gamma_{2l} = (a-\ell)(\tilde{I}_{k_l}+N)$ and
\begin{equation}
    \Delta_l = \sum_{l=1}^\ell\frac{(-\Delta_{k_2(l)})}{(2^{\underline{R'_{k_l}}} - 1) \gamma_{1l}} (a - K_2(l) - 1) (a - \ell) + P_{LI}\sum_{l=1}^\ell \frac{(a - K_2(l) - 1) (a - \ell)}{\gamma_{1l}}.
    \label{eq: deltasm}
\end{equation} 
Then the following inequalities for $1 \leq l \leq \ell$ are sufficient for (\ref{eq:lastiq})
\begin{equation}
    \frac{\gamma_{2l}}{\gamma_{1l}} \beta_{k_{l}, \alpha'_l}^{-1}  \leq \beta_{k_{l}, \alpha_l}^{-1} + \Delta_l.
\label{eq:cond_proof}
\end{equation}

\end{proof}

The following remarks hold for all $\frac{\gamma_{2l}}{\gamma_{1l}}$ with $1 \leq l \leq \ell$, however, we only state it for a single user. 

\begin{remark}
From (\ref{eq:ratio_gamma}), it can be inferred that at least $a > K_2 + 1$ antennas are required.  
\begin{equation}
    \frac{\gamma_{21}}{\gamma_{11}} = \frac{a-1}{a-K_2-1}\frac{\tilde{I}_{k_1}+N}{I_{k_1}+N}
\label{eq:ratio_gamma}
\end{equation}
Let $a = K_2 + N$, where $N > 1$, then, $\frac{\gamma_{21}}{\gamma_{11}}$ can be written as, 
\begin{equation}
    \frac{\gamma_{21}}{\gamma_{11}} = \left(\frac{K_2}{N - 1} + 1\right) \left(\frac{\tilde{I}_{k_1}+N}{I_{k_1}+N}\right).
\label{eq:ratio_gamma_update}
\end{equation}

Since, $I_{k_1} > \tilde{I}_{k_1}$, the ratio of the interference is bounded, i.e, $0 < \left(\frac{\tilde{I}_{k_1}+N}{I_{k_1}+N}\right) \leq 1$, and therefore, $\frac{\gamma_{21}}{\gamma_{11}}$ is bounded, i.e., 
\begin{equation}
    0 \leq \frac{\gamma_{21}}{\gamma_{11}} \leq \left(\frac{K_2}{N - 1} + 1\right).
\end{equation}
\end{remark}

\begin{remark}
Under the assumption that $\beta_{k_1, \alpha'} \leq \beta_{k_1, \alpha}$, it can be observed in (\ref{eq:cond_proof}) that $\beta_{k_1, \alpha'} ^{-1} \geq \beta_{k_1, \alpha}^{-1}$. Thus, for (\ref{eq:cond_proof}) to hold,  $\Delta$ in (\ref{eq:cond_proof}) must be positive, i.e, 
\begin{eqnarray}
    \Delta & > & 0  \\
    \frac{(-\Delta_{k_2} )}{(2^{\underline{R'_{1}}} - 1) \gamma_{11}} (a - K_2 - 1) (a -1) + P_{LI} \frac{(a - K_2 - 1) (a -1)}{\gamma_{11}}& > & 0. \label{eq:cond_bound}
\end{eqnarray}
From (\ref{eq:cond_bound}), it can be inferred that
\begin{equation}
    \Delta_{k_2} < P_{LI}(2^{\underline{R'_{1}}} - 1). 
\end{equation}
\end{remark}

\begin{proof}
\emph{ii)} \textbf{No Spatial Multiplexing}: Similarly, for the given minimum rate requirements $\underline{R}_{k_l}$, the required transmit power for user $k_l$ with PRBs $\alpha_l$ in BS $m_1$ can be computed as 
\begin{equation}
    p_{k_l, \alpha_l} =  \frac{(2^{\underline{R'_l}} - 1)}{a - 1} \frac{(I_{k_l} + N)}{\beta_{k_l, \alpha_l}},
\end{equation}
with $\underline{R'_{k_l}} = \frac{\underline{R_{k_l}}}{\bar{b} \cdot N_m}$ and maximum available PRBs $N_m$. The total power consumption of BS $m_1$, including the load independent power ($P_{LI}$), is calculated as
\begin{equation}
    P_{m_1} = \frac{N_m}{\eta_{PA}} \sum_{l=1}^\ell p_{k_l, \alpha_l} + P_{LI}. 
\end{equation}

Under the assumption that BS $m_2(l)$ has the second best channel $\beta_{k_l, \alpha'_l}$ for user $k_l$, if all users $\{k_1,...,k_\ell\}$ are assigned to their second best BS $m_2(1),...,m_2(\ell)$, $P_{m_1}$ amount of power/energy can be saved. However, to support user $k_1,...,k_\ell$ by BS $m_2(1),...,m_2(\ell)$, additional power is needed because  \emph{a)} users $k_1,...,k_\ell$ need to be served and \emph{b)} users served by $m_2(1),...,m_2(\ell)$ need more power.
Assuming $K_2(l) > 0$ users, $\mu(m_2(l)) = \{k_{l1}, \dots, k_{lK_2}\}$, are already assigned to $m_2(l)$, when user $k_l$ joins BS $m_2(l)$ with PRBs $\alpha'_l$ assigned, the required transmit power can be computed as
\begin{equation}
    p_{k_l, \alpha'_l}  = \frac{(2^{\underline{R'_{k_l}}} - 1)(\tilde{I}_{k_l} + N)}{\beta_{k_l, \alpha'_l}(a - 1)},
\end{equation}
with $\underline{R'_{k_l}} = \frac{\underline{R_{k_l}}}{\bar{b} \cdot \alpha'_{l}}$, where $\alpha'_l$ represents the number of PRBs assigned to user $k_l$. It is observed that $\tilde{I}_{k_l} < I_{k_l}$ and $\beta_{k_l, \alpha'_l} \leq \beta_{k_l, \alpha_l}$, and the difference between the transmit powers after and before user $k_l$ joins BS $m_2(l)$ can be computed as
\begin{equation}
    \Delta_{k_{2l}} = p'_{k_{2l}, \alpha'_{2l}} - p_{k_{2l}, \alpha'_{2l}},  
\end{equation}
where $p_{m_2(l), k_{2l}} = p_{k_{2l}, \alpha'_{2l}}$ is the required transmit power before user $k_l$ joins BS $m_2(l)$ and is computed as 
\begin{equation}
   p_{k_{2l}, \alpha'_{2l}} = \frac{(2^{\underline{R'_{k_{2l}}}} - 1)(I_{k_{2l}} + N)}{\beta_{k_{2l}, \alpha'_{2l}}(a - 1)}, 
\end{equation}
with $\underline{R'_{k_{2l}}} = \frac{\underline{R_{k_{2l}}}}{\bar{b} \cdot \alpha'_{2l}}$ and $p'_{m_2(l), k_{2l}} = p'_{k_{2l}, \alpha'_{2l}}$ is the required transmit power after user $k_l$ joins BS $m_2(l)$ and is computed as
\begin{equation}
    p'_{k_{2l}, \alpha'_{2l}} = \frac{(2^{\underline{R''_{k_{2l}}}} - 1)(\tilde{I}_{k_{2l}} + N)}{\beta_{k_{2l}, \alpha'_{2l}}(a - 1)},
\end{equation}
with $\underline{R''_{k_{2l}}} = \frac{\underline{R_{k_{2l}}}}{\bar{b} \cdot \tilde{\alpha'}_{2l}}$
Then, 
\begin{eqnarray}
    \Delta_{k_{2l}} = \frac{1}{\beta_{k_{2l}, \alpha'_{2l}}(a - 1)} \left((2^{\underline{R''_{k_{2l}}}} - 1)(\tilde{I}_{k_{2l}} + N) - (2^{\underline{R'_{k_{2l}}}} - 1)(I_{k_{2l}} + N)\right), \\ 
     = \frac{1}{\beta_{k_{2l}, \alpha'_{2l}}(a - 1)} \left(2^{\underline{R''_{k_{2l}}}}(\tilde{I}_{k_{2l}} + N) - 2^{\underline{R'_{k_{2l}}}}(I_{k_{2l}} + N) + I_{k_{2l}} - \tilde{I}_{k_{2l}} \right),
\end{eqnarray}

and the total difference in power accounting for all the users assigned to BS $m_2$ can be calculated as 
\begin{equation}
    \sum_{k \in \mu(m_2(l))} \Delta_k + \Delta_{k_l}
\end{equation}
with
\begin{equation}
    \Delta_{k_l} = \frac{1}{(a-1)}\left(\frac{(2^{\underline{R''_{k_{l}}}} - 1)(\tilde{I}_{k_{l}} + N)}{\beta_{k_{l}, \alpha'_{l}}}  - \frac{(2^{\underline{R'_{k_l}}} - 1)(I_{k_{l}} + N)}{\beta_{k_{l}, \alpha_{l}}} \right). 
\end{equation}

Assuming the load independent power consumption for both the BSs is the same, the difference of power consumption before and after re-matching can be obtained and com[pared to zero as 
\begin{alignat}{2}
    && \sum_{l=1}^\ell\sum_{k \in \mu(m_2(l))} \Delta_k + \Delta_{k_l} - P_{LI} & \leq  0 , \\
    &\implies 
    &\sum_{l=1}^\ell\Delta_{k_2(l)} + \frac{1}{(a-1)}\left(\frac{(2^{\underline{R''_{k_l}}} - 1)(\tilde{I}_{k_{l}} + N)}{\beta_{k_{l}, \alpha'_{l}}}  - \frac{(2^{\underline{R'_{k_l}}} - 1)(I_{k_{l}} + N)}{\beta_{k_{l}, \alpha_l}} \right) - P_{LI}
    & \leq  0, \\
    &\implies
    &\sum_{l=1}^\ell (2^{\underline{R''_{k_l}}} - 1)(\tilde{I}_{k_{l}} + N)\beta_{k_{l}, \alpha'_l}^{-1}  - (2^{\underline{R'_{k_l}}} - 1)(I_{k_{l}} + N)\beta_{k_{l}, \alpha_l}^{-1} 
    & \leq \label{eq:lastiqnsm} \\ 
    &&\left(\sum_{l=1}^\ell(-\Delta_{k_2(l)}) + P_{LI}\right)(a-1), \nonumber
\end{alignat}
where $\Delta_{k_2(l)} = \sum_{k \in \mu(m2(l))} \Delta_k$. Let $\gamma_{1l} = (2^{\underline{R'_{k_l}}} - 1)(I_{k_l}+N)$,  $\gamma_{2l} = (2^{\underline{R''_{k_l}}} - 1)(\tilde{I}_{k_l}+N)$ and 
\begin{equation}
    \Delta_l = \sum_{l=1}^\ell\frac{(-\Delta_{k_2(l)})}{\gamma_{1l}}(a-1) +\frac{P_{LI}(a-1)}{\sum_{l=1}^\ell \gamma_{1l}} .
    \label{eq: deltansm}
\end{equation}
Then the following inequalities for $1 \leq l \leq \ell$ are sufficient for (\ref{eq:lastiqnsm})
\begin{equation}
    \frac{\gamma_{2l}}{\gamma_{11}} \beta_{k_{l}, \alpha'_l}^{-1}  \leq \beta_{k_{l}, \alpha_l}^{-1} + \Delta_l.
\label{eq:cond_proof_1}
\end{equation}

\end{proof}

The following remarks hold for all $\frac{\gamma_{2l}}{\gamma_{1l}}$ with $1 \leq l \leq \ell$, however, we only state it for a single user. 

\begin{remark}
From (\ref{eq:ratio_gamma_1}), it can be inferred that at least $\underline{R'_1} > 0$.  
\begin{equation}
    \frac{\gamma_{21}}{\gamma_{11}} = \frac{(2^{\underline{R''_{1}}} - 1)}{(2^{\underline{R'_{1}}} - 1)}\frac{(\tilde{I}_{k_1}+N)}{(I_{k_1}+N)}.
\label{eq:ratio_gamma_1}
\end{equation}
\end{remark}

\begin{remark}
Under the assumption that $\beta_{k_1, \alpha'} \leq \beta_{k_1, \alpha}$, it can be observed in (\ref{eq:cond_proof_1}) that $\beta_{k_1, \alpha'} ^{-1} \geq \beta_{k_1, \alpha}^{-1}$. Thus, for (\ref{eq:cond_proof_1}) to hold,  $\Delta$ in (\ref{eq:cond_proof_1}) must be positive, i.e, 
\begin{eqnarray}
    \Delta & > & 0  \\
    \frac{(-\Delta_{k_2} + P_{LI})}{\gamma_{11}} (a -1) & > & 0. \label{eq:cond_bound_1}
\end{eqnarray}
From (\ref{eq:cond_bound_1}), it can be inferred that
\begin{equation}
    \Delta_{k_2} < P_{LI}. 
\end{equation}
\end{remark}

\section{Proof of Theorem~\ref{theo:2}}
\label{sec: proof_theo2}
\begin{proof}

The proof for both NSM and SM is done by contra-diction as follows:

\emph{i)} \textbf{No Spatial Multiplexing}: Let $\tilde{p}_m$ be the transmit power allocated to the users in a cell with fixed interference. For any $\tilde{p}_m > 0$, $\sum_{k \in \mu(m)}\alpha_k(\tilde{p}_m)$ can be expressed as
\begin{equation}
    \sum_{k \in \mu(m)}\alpha_k(\tilde{p}_m) = \sum_{k \in \mu(m)} \frac{\underline{R_k}}{\bar{b} \log_2 \left(1 + \frac{(a_m - |\mu(\alpha)|)\beta_k \tilde{p}_m}{I_k + N}\right)},
    \label{eq:prbs}
\end{equation}
assuming equal power allocation per PRB. If $\tilde{p}_m < p_m^*$, it can be observed from (\ref{eq:prbs}), that 
\begin{eqnarray}
    \sum_{k \in \mu(m)} \alpha_k(\tilde{p}_{m}) > N_m 
    \label{eq:cNSM1}, 
\end{eqnarray}
where $N_m$ is the maximum available PRBs. However, since the system is constrained by the maximum number of available PRBs, any $\tilde{p}_m < p_m^*$ leads to outages because the minimum rate requirement of the users is not satisfied.   

For any $\tilde{p}_m > p_m^*$, it can be observed that 
\begin{eqnarray}
    \sum_{k \in \mu(m)} \alpha_k(\tilde{p}_{m}) < N_m, 
    \label{eq:cNSM2}
\end{eqnarray}
leading to a lower EE due to an increase in the total power consumption of the BS, i.e., 
\begin{eqnarray}
    P_m(\tilde{p}_m) > P_m(p_m^*).
    \label{eq:ctotpow}
\end{eqnarray}

To show that (\ref{eq:ctotpow}) holds, it suffices to show that updating the allocated transmit power for a single user results in a lower EE. Thus, between the vectors $\boldsymbol{\tilde{p}_m}$ and $\boldsymbol{p_m^*} $, the transmit power for only one user is updated, i.e., $\boldsymbol{\tilde{p}_m} \rightarrow \boldsymbol{p_m^*} = [p_{m1}^*, \tilde{p}_{m2}, \dots, \tilde{p}_{mk}]$ and (\ref{eq:cNSM2}) and (\ref{eq:cNSM}) can be expressed as\footnote{The subscript for the BS $m$ has been dropped for simplicity.}
\begin{eqnarray}
    \alpha_1(\tilde{p}_{1}) + \sum_{k=2}^K \alpha_k(\tilde{p}_{k}) &=& N_m - \delta \label{eq:cNSM_update1}\\
    \alpha_1(p_{1}^*) + \sum_{k=2}^{K} \alpha_k(\tilde{p}_{k}) &=& N_m, 
    \label{eq:cNSM_update2}
\end{eqnarray}
respectively, where $\delta \geq 0$ represents the change in the number of PRBs due to an increase in the allocated transmit power for a single user. Subtracting  (\ref{eq:cNSM_update1}) from (\ref{eq:cNSM_update2}), $\delta$ can be obtained as 
\begin{equation}
    \delta = \alpha_1(p_{1}^*) - \alpha_1(\tilde{p}_{1})
    \label{eq:cNSM_delt}
\end{equation}
and the total power consumption of the BS\footnote{It should be noted that the terms independent of PRBs $\alpha_k$ and transmit power $p_k$ can be treated as constants and therefore, have been omitted here.} with the optimal power $p_1^*$ can be computed as
\begin{eqnarray}
    P_m(\boldsymbol{p_m}^*) &=& \frac{1}{\eta_{PA}}\alpha_1(p_1^*)p_1^* + C \alpha_1(p_1^*) + \frac{1}{\eta_{PA}}\sum_{k = 2}^K \alpha_k(\tilde{p}_k) \tilde{p}_{k} + C \sum_{k=2}^K \alpha_k(\tilde{p}_k) \\
      &=& \alpha_1(p_1^*) \left(\frac{1}{\eta_{PA}}p_1^* + C\right) + \sum_{k = 2}^K \alpha_k(\tilde{p}_k) \left(\frac{1}{\eta_{PA}} \tilde{p}_{k} + C\right),   \label{eq:power_pstar}
\end{eqnarray}
whereas, the total power consumption of the BS with $\boldsymbol{\tilde{p}_m}$ can be computed as
\begin{eqnarray}
    P_m(\boldsymbol{\tilde{p}_m}) &=&  \frac{1}{\eta_{PA}} \sum_{k = 1}^K \alpha_k(\tilde{p}_k) \tilde{p}_k + C \sum_{k=1}^K \alpha_k(\tilde{p}_k) \\
    &=&  \alpha_1(\tilde{p}_1) \left(\frac{1}{\eta_{PA}}\tilde{p}_1 + C\right) + \sum_{k = 2}^K \alpha_k(\tilde{p}_k) \left(\frac{1}{\eta_{PA}} \tilde{p}_{k} + C\right) \\
    &=&  (\alpha_1(p_1^*) - \delta) \left(\frac{1}{\eta_{PA}}\tilde{p}_1 + C\right) + \sum_{k = 2}^K \alpha_k(\tilde{p}_k) \left(\frac{1}{\eta_{PA}} \tilde{p}_{k} + C\right) \label{eq:power_ptilde}, 
\end{eqnarray}
where, in (\ref{eq:power_pstar}) and (\ref{eq:power_ptilde}), $C = D_{m,1}a_m \left(\sum_{k}N_k \right)$ as shown in (\ref{eq:ecm}). From (\ref{eq:ctotpow}), after simplifying and using (\ref{eq:cNSM_delt}), it can be seen that
\begin{alignat}{2}
    &&\alpha_1(p_1^*) \left(\frac{1}{\eta_{PA}}\tilde{p}_1 + C\right) - \alpha_1(p_1^*) \left(\frac{1}{\eta_{PA}}p_1^* + C\right) - \delta \left(\frac{1}{\eta_{PA}}\tilde{p}_1 + C\right) &\geq 0 \\
    &\implies
    &\alpha_1(p_1^*) (\tilde{p}_1 - p_1^*) - \delta (\tilde{p}_1 + C\eta_{PA}) &\geq 0 \\
    &\implies 
    &\tilde{p}_1 (\alpha_1(\tilde{p}_1) + \delta) - \alpha_1(p_1^*)p_1^* - \delta \tilde{p}_1 + \delta C\eta_{PA} &\geq 0 \\
    &\implies
    &\alpha_1(\tilde{p}_1)\tilde{p}_1 - \alpha_1(p_1^*)p_1^* &\geq \delta C\eta_{PA}, 
\end{alignat}
which states that the difference between load-dependent part with $\tilde{p}$ and $p^*$, i.e., $\alpha_1(\tilde{p}_1)\tilde{p}_1 - \alpha_1(p_1^*)p_1^*$ must be greater than or equal to the load-independent part $\delta C\eta_{PA}$. 

\emph{ii)} \textbf{Spatial Multiplexing}: Similarly, it can be shown for SM that the optimal power allocation is obtained by finding the power $\boldsymbol{p^*}$ such that each user is assigned maximum available PRBs. 

From (\ref{eq:prbs_sm}), it can be observed that any $\tilde{p}_m < p_m^*$ leads to an increase in the number of PRBs exceeding maximum available PRBs $N_m$ 
\begin{equation}
    \alpha_k(\tilde{p}_m) =  \frac{\underline{R_k}}{\bar{b} \log_2 \left(1 + \frac{(a_m - |\mu(\alpha)|)\beta_k \tilde{p}_m}{I_k + N}\right)}
    \label{eq:prbs_sm}
\end{equation}
However, since no more than maximum available PRBs can be assigned to a user, $\tilde{p}_m < p_m^*$ leads to outages because the minimum rate requirement of the users with $\tilde{p}_m$ cannot be satisfied.  

For any $\tilde{p}_m > p_m^*$, it can be observed that 
\begin{eqnarray}
     (\alpha_1(\tilde{p}_{1}), \dots, \alpha_k(\tilde{p}_k))_{k \in \mu(m)} < N_m, 
    \label{eq:cSM1}
\end{eqnarray}
leading to a lower EE due to an increase in the total power consumption of the BS, i.e., 
\begin{eqnarray}
    P_m(\tilde{p}_m) > P_m(p_m^*).
    \label{eq:ctotpow_sm}
\end{eqnarray}
To show that (\ref{eq:ctotpow_sm}) holds, it suffices to show that updating the allocated transmit power for a single user results in a lower EE. Thus, between the vectors $\boldsymbol{\tilde{p}_m}$ and $\boldsymbol{p_m^*} $, the transmit power for a only one user is updated, i.e., $\boldsymbol{\tilde{p}_m} \rightarrow \boldsymbol{p_m^*} = [p_{m1}^*, \tilde{p}_{m2}, \dots, \tilde{p}_{mk}]$ and (\ref{eq:cSM1}) and (\ref{eq:cSM}) for this user can be expressed as\footnote{The subscript for the BS $m$ has been dropped for convenience.}
\begin{eqnarray}
    \alpha_1(\tilde{p}_{1})  &=& N_m - \delta \label{eq:cSM_update1}\\
    \alpha_1({p}_{1}^*)  &=& N_m, 
    \label{eq:cSM_update2}
\end{eqnarray}
respectively, where $\delta \geq 0$ represents the change in the number of PRBs due to an increase in the allocated transmit power for a single user. Subtracting  (\ref{eq:cSM_update1}) from (\ref{eq:cSM_update2}), $\delta$ can be obtained as 
\begin{equation}
    \delta = \alpha_1(p_{1}^*) - \alpha_1(\tilde{p}_{1})
    \label{eq:cSM_delt}
\end{equation}
and the total power consumption of the BS\footnote{It should be noted that the terms independent of PRBs $\alpha_k$ and transmit power $p_k$ can be treated as constants and therefore, have been omitted here.} with the optimal power $p_1^*$ can be computed as
\begin{eqnarray}
    P_m(\boldsymbol{p_m}^*) &=& \frac{1}{\eta_{PA}}\alpha_1(p_1^*)p_1^* + C \alpha_1(p_1^*) + \frac{1}{\eta_{PA}}\sum_{k = 2}^K \alpha_k(\tilde{p}_k) \tilde{p}_{k} + C \sum_{k=2}^K \alpha_k(\tilde{p}_k) \\
      &=& \alpha_1(p_1^*) \left(\frac{1}{\eta_{PA}}p_1^* + C\right) + N_m(K-1) \left(\frac{1}{\eta_{PA}} \tilde{p}_{k} + C\right),   \label{eq:power_pstar_sm}
\end{eqnarray}
whereas, the total power consumption of the BS with $\boldsymbol{\tilde{p}_m}$ can be computed as
\begin{eqnarray}
    P_m(\boldsymbol{\tilde{p}_m}) &=&  \frac{1}{\eta_{PA}} \sum_{k = 1}^K \alpha_k(\tilde{p}_k) \tilde{p}_k + C \sum_{k=1}^K \alpha_k(\tilde{p}_k) \\
    &=&  \alpha_1(\tilde{p}_1) \left(\frac{1}{\eta_{PA}}\tilde{p}_1 + C\right) + \sum_{k = 2}^K \alpha_k(\tilde{p}_k) \left(\frac{1}{\eta_{PA}} \tilde{p}_{k} + C\right) \\
    &=&  (\alpha_1(p_1^*) - \delta) \left(\frac{1}{\eta_{PA}}\tilde{p}_1 + C\right) + N_m (K-1) \left(\frac{1}{\eta_{PA}} \tilde{p}_{k} + C\right) \label{eq:power_ptilde_sm}, 
\end{eqnarray}
where, in (\ref{eq:power_pstar_sm}) and (\ref{eq:power_ptilde_sm}), $C = D_{m,1}a_m \left(\sum_{k}N_k \right)$ as shown in (\ref{eq:ecm}). From (\ref{eq:ctotpow_sm}), after simplifying and using (\ref{eq:cSM_delt}), it can be seen that
\begin{alignat}{2}
    &&\alpha_1(p_1^*) \left(\frac{1}{\eta_{PA}}\tilde{p}_1 + C\right) - \alpha_1(p_1^*) \left(\frac{1}{\eta_{PA}}p_1^* + C\right) - \delta \left(\frac{1}{\eta_{PA}}\tilde{p}_1 + C\right) &\geq 0 \\
    &\implies
    &\alpha_1(p_1^*) (\tilde{p}_1 - p_1^*) - \delta (\tilde{p}_1 + C\eta_{PA}) &\geq 0 \\
    &\implies
    &\tilde{p}_1 (\alpha_1(\tilde{p}_1) + \delta) - \alpha_1(p_1^*)p_1^* - \delta \tilde{p}_1 + \delta C\eta_{PA} &\geq 0 \\
    &\implies
    &\alpha_1(\tilde{p}_1)\tilde{p}_1 - \alpha_1(p_1^*)p_1^* &\geq \delta C\eta_{PA}, 
\end{alignat}
which states that the difference between load-dependent part with $\tilde{p}$ and $p^*$, i.e., $\alpha_1(\tilde{p}_1)\tilde{p}_1 - \alpha_1(p_1^*)p_1^*$ must be greater than or equal to the load-independent part $\delta C\eta_{PA}$. 
\end{proof} 

\section{Large-scale Fading Computation}
\label{sec:large_scale_fading}

Highly accurate channel estimates can be determined using ray tracing \cite{doi:https://doi.org/10.1002/9781118410998.ch10}. In this work, we use the in-house developed ray tracing based large-scale fading predictor, called Femto Predictor (\textit{FemtoPred}), which was developed at the Institute for Communications Technology at TU Braunschweig. To describe an outdoor model with an arbitrary number of UEs, the \textit{FemtoPred} processes 3D building data accounting for reflection, diffraction and transmission, along with different antenna diagrams. The \textit{FemtoPred} starts by creating pairs between cell sites and map pixels, and deriving the reflected, diffracted and transmitted paths, even if a LoS path is detected. Then, the channels are calculated considering path losses, shadow and multi-path fading, while accounting for LoS and non line-of-sight (NLoS).

\subsubsection {Large-scale Fading}

Within our scenario, a macroscopic approach was adopted, where, every pixel of the map represents a large-scale fading value. Since we are simulating multiple frequencies ($700$ MHz and $2.6$ GHz), ray tracing is well-suited for simulating different wave lengths of different frequencies accounting for the distance between the $T_x$ and $R_x$ in case of a LOS. 

The large-scale fading prediction inside the \textit{FemtoPred} can be divided into three steps \cite{ThoR}: firstly, via geometrical approach, the image method is used to find the only qualified rays to speed up the calculations. Secondly, for each ray electromagnetic calculations are performed to determine the path loss. In addition, the angel-of-arrival (AoA) at the receiver and the angle-of departure (AoD) at the transmitter are derived. Finally, to calculate the antenna pattern, the 3D antenna gain is multiplied by the complex received power. 

For NLOS, the large-scale fading calculations consider multiple components including diffraction, reflection and transmission, as described below.

\paragraph{Diffraction}

The knife edge model by Deygout \cite{1138719} was used to calculate the buildings' roof diffraction. For a direct ray path, all the intersected surfaces are determined on the assumption that a diffracted edge always belongs to the surface of the intersected ray path. Therefore, in case of a single obstacle intersecting the path of the ray, the diffraction $C$ is computed as \cite{doi:https://doi.org/10.1002/9781118410998.ch12}:

\begin{eqnarray}
    C_{K,E} (dB)= 6.9 + 20 \log(\sqrt{(v-0.1)^2 + 1} + v - 0.1)
\label{eq: CKE1}
\end{eqnarray}

where, $v$ is the Fresnel diffraction. As shown in Figure\footnote{Source: LTE-Advance and Next Generation Wireless Networks} {\ref{fig:Fresnel_diffraction}}, the Fresnel diffraction is dependant on the height $h$ of the obstacle that intercepts the ray, the distance between the transmitter $T_x$ and the obstacle $r1$ (along the LoS) and the distance between the obstacle $r2$ and the receiver $R_x$ (along the LoS). 

\begin{figure}[htp]
    \centering    \includegraphics[width=0.5\textwidth]{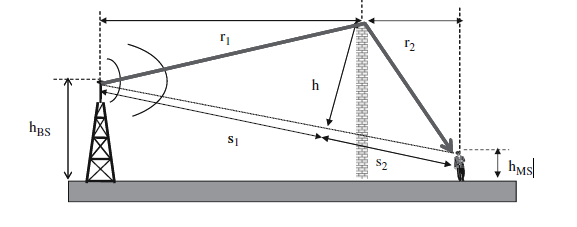}
    \caption{Calculating Fresnel diffraction}
    \label{fig:Fresnel_diffraction}
\end{figure}

Thus, the Fresnel diffraction is computed as, 
\begin{eqnarray}
   v = h \sqrt{\frac{2}{\lambda} \frac{s1 + s2}{r1r2}}.
\label{eq: FresnelDiffraction}
\end{eqnarray}
where, and $\lambda$ is the wavelength and $s1 = \sqrt{h^2 + r1^2}$ and $s2 = \sqrt{h^2 + r2^2}$ are the Pythagorean distance between the $T_x$ and $R_x$. 

For multiple obstacles that intersect the path of the ray, the Deygout method describes the following algorithm:

\begin{itemize}
    \item {Step 1 :} The main obstacle between the $T_x$ and the $R_x$ is calculated based on the highest $v$ values from (\ref{eq: FresnelDiffraction}), where only a single knife edge is considered in  between and the diffraction loss\footnote{In this implementation, the diffraction loss was also multiplied by $k$ which is the \emph{empirical correction factor} with values ranging between $0.2$ to $0.5$.} is $C_2(v_2)$.
    \item {Step 2 :} The filtering procedure of the qualified points to the left $\Vec{C}(v_{1,left})$ and the right of the ray $\Vec{C}(v_{3,right})$.
    \item {Step 3 :} If more obstacles are still available, then Step $2$ is repeated.
    
In case of three obstacles between the BS and the mobile station, total diffraction loss, using the Deygout method, can be  calculated as,

\begin{eqnarray}
   C_{d} (dB) = C_2 (v_2) + C_1 {(v_{1,left})} + C_3 {(v_{3,right})}
\label{eq: CKE2}
\end{eqnarray}

\end{itemize}

\paragraph{Reflection and Transmission}

In our simulator, the reflection and transmission of the signals is dependant on the signal frequency. For this to be as accurate as possible, a dielectric layer needs to be implemented based on the material of the buildings inside the simulated environment (in this work, concrete buildings).

The direction of the incidence can be defined as the angle between the wave vector and the normal to the interface \cite{doi:https://doi.org/10.1002/9781118410998.ch2}. The angle of incident ($\Theta_e$), transmitted ($\Theta_t$) and reflected ($\Theta_r$) wave  can be defined as,

\begin{eqnarray}
   \frac{\sin{\Theta_t} }{\sin{\Theta_e}} = \frac{\sqrt{\epsilon_1}}{\sqrt{\epsilon_2}}
\label{eq: R&T2}
\end{eqnarray}

where, $\Theta_r = \Theta_e$, and $\epsilon_1$, $\epsilon_2$ are the dielectric constants of the each side of the material.

To compute the magnitude of the wave, a distinction between whether the electric field vector is parallel to the interface of the dielectrics (TE wave) and whether the magnetic field is parallel to the interface transversal magnetic field (TM wave) is necessary, where  $\rho_{TE}$, $T_{TE}$, $\rho_{TM}$, and $ T_{TM}$ can be computed as,

\begin{eqnarray}
   \rho_{TE} = \frac{\sqrt{\epsilon_1} \cos{\Theta_e} - \sqrt{\epsilon_2} \cos{\Theta_t} } 
   {\sqrt{\epsilon_1} \cos{\Theta_e} + \sqrt{\epsilon_2} \cos{\Theta_t}}
\label{eq:P_TE_1}
\end{eqnarray}

\begin{eqnarray}
   T_{TE} = \frac{2 \sqrt{\epsilon_1} \cos{\Theta_e} } 
   {\sqrt{\epsilon_1} \cos{\Theta_e} + \sqrt{\epsilon_2} \cos{\Theta_t}}
\label{eq: T_TE_1}
\end{eqnarray}

\begin{eqnarray}
   \rho_{TM} = \frac{\sqrt{\epsilon_2} \cos{\Theta_e} - \sqrt{\epsilon_1} \cos{\Theta_t} } 
   {\sqrt{\epsilon_2} \cos{\Theta_e} + \sqrt{\epsilon_1} \cos{\Theta_t}}
\label{eq:P_TE_2}
\end{eqnarray}

\begin{eqnarray}
   T_{TM} = \frac{2 \sqrt{\epsilon_1} \cos{\Theta_e} } 
   {\sqrt{\epsilon_2} \cos{\Theta_e} + \sqrt{\epsilon_1} \cos{\Theta_t}}
\label{eq: T_TE_2}
\end{eqnarray}

Conclusively, the overall reflection and transmission values can be calculated for a dielectric layer of thickness $d_{layer}$ as,

\begin{eqnarray}
  T = \frac{T_1T_2 e ^ {-ja}}{1 + \rho_1 \rho_2 e ^ {-2ja}},
\label{eq: T}
\end{eqnarray}

\begin{eqnarray}
  \rho = \frac{\rho_1 + \rho_2 e ^ {-ja} }{1 + \rho_1 \rho_2 e ^ {-2ja}},
\label{eq: rho_1}
\end{eqnarray}

\begin{eqnarray}
  \alpha = \frac{2\pi}{\lambda} \sqrt{\epsilon_2} d_{layer} \cos({\Theta_t})
\label{eq: rho_2}
\end{eqnarray}

where $\rho_1, \rho_2$, $T_1$, and $T_2$ are the reflection and transmission coefficients respectively and $\alpha$ is the electrical length of the layer.
\subsubsection{Antenna Modeling} 

In our simulation, directional antenna patterns have been implemented for both frequencies. For the $700$ MHz band, a $3$GPP antenna model \cite{3GPP} was generated, where, the antenna element $A_{E}$ consists of two main components, i.e., vertical and horizontal propagation patterns.

The antenna diagram can be modeled based on (\ref{eq: VPattern}) and (\ref{eq: HPattern}), 

\begin{eqnarray}
    A_{E,V}(\theta)\text{(dB)} = -\min\Bigg[12 \Bigg(\frac { \theta - 90 ^{\circ}}{\theta_{3dB}}\Bigg)^2, SLA_v\Bigg]
\label{eq: VPattern}
\end{eqnarray}

\begin{eqnarray}
    A_{E,H}(\phi)\text{(dB)} = -\min\Bigg[12 \Bigg(\frac { \phi''}{\phi_{3dB}}\Bigg)^2, A_m\Bigg]
\label{eq: HPattern}
\end{eqnarray}
where, $A_{E,V}$ and $A_{E,H}$ represent the vertical and horizontal propagation pattern, respectively. A $3$D antenna element pattern can be acquired by combining the vertical and horizontal patterns and adding the maximum antenna gain as,

\begin{eqnarray}
A(\theta,\phi)\text{(dB)}= G_{E,max} -\min\bigg\{ - \bigg[A_{E,V}(\theta)+A_{E,H}(\phi) \bigg],A_m \Bigg\}
\label{eq: CPattern}
\end{eqnarray}
where, $\theta$ is the zenith angle ([0 , 180] degrees), $\phi$ is the azimuth angle ([-180 , 180] degrees), $\phi_{3dB} = 65^{\circ}$ is the $3$dB beamwidth of the horizontal antenna pattern, $\theta_{3dB}$ = $15^{\circ}$ is the  $3$dB beamwidth of the vertical antenna pattern, $SLA_v$ = 30 is the sidelobe attenuation, $A_m$ is the maximum attenuation of the main lobe, and $G_{E,max} = 8$ dBi is the maximum antenna element gain\footnote{These parameter values have been specified for the $700$ MHz antenna diagram.}.  
 
For the 2.6 GHz band a single band sector "Commscope (HWXX-6516DS-VTM2600)" antenna diagram pattern was chosen.

\bibliographystyle{IEEEtran}
\bibliography{IEEEabrv,paper}

\end{document}